\documentclass[oneside,english,reqno]{amsart}
\usepackage[T1]{fontenc}
\usepackage[latin9]{inputenc}
\usepackage{babel}
\usepackage{enumitem}
\usepackage{amstext}
\usepackage{amsthm}
\usepackage{hyperref}
\usepackage{breakurl}
\usepackage{color}

\makeatletter
\numberwithin{equation}{section}
\numberwithin{figure}{section}
\theoremstyle{plain}
\newtheorem{thm}{\protect\theoremname}
  \theoremstyle{remark}
  
  \theoremstyle{plain}
  \newtheorem{lem}[thm]{\protect\lemmaname}
 \newlist{casenv}{enumerate}{4}
 \setlist[casenv]{leftmargin=*,align=left,widest={iiii}}
 \setlist[casenv,1]{label={{\itshape\ \casename} \arabic*.},ref=\arabic*}
 \setlist[casenv,2]{label={{\itshape\ \casename} \roman*.},ref=\roman*}
 \setlist[casenv,3]{label={{\itshape\ \casename\ \alph*.}},ref=\alph*}
 \setlist[casenv,4]{label={{\itshape\ \casename} \arabic*.},ref=\arabic*}

\makeatother

  \providecommand{\lemmaname}{Lemma}
  \providecommand{\remarkname}{Remark}
 \providecommand{\casename}{Case}
\providecommand{\theoremname}{Theorem}

\begin{document}

\title[Blow-up profile of ground states of boson stars]{On blow-up profile of ground states of boson stars with external potential}

\author{Dinh-Thi Nguyen}
\address{Dinh-Thi Nguyen, Mathematisches Institut, Ludwig-Maximilans-Universit\"at M\"unchen, Theresienstr. 39, 80333 M\"unchen, Germany.} 
\email{nguyen@math.lmu.de}

\maketitle

\begin{abstract}
We study minimizers of the pseudo-relativistic Hartree functional $\mathcal{E}_{a}(u):=\|(-\Delta+m^{2})^{1/4}u\|_{L^{2}}^{2}-\frac{a}{2}\int_{\mathbb{R}^{3}}(\left|\cdot\right|^{-1}\star |u|^{2})(x)|u(x)|^{2}{\rm d}x+\int_{\mathbb{R}^{3}}V(x)|u(x)|^{2}{\rm d}x$ under the mass constraint $\int_{\mathbb{R}^3}|u(x)|^2{\rm d}x=1$. Here $m>0$ is the mass of particles and $V\geq 0$ is an external potential. We prove that minimizers exist if and only if $a$ satisfies $0\leq a<a^{*}$, and there is no minimizer if $a\geq a^*$, where $a^*$ is called the \emph{Chandrasekhar limit}. When $a$ approaches $a^*$ from below, the blow-up behavior of minimizers is derived under some general external potentials $V$. Here we consider three cases of $V$: trapping potential, i.e. $V\in L_{{\rm loc}}^{\infty}(\mathbb{R}^3)$ satisfies $\lim_{|x|\to \infty}V(x)=\infty$; periodic potential, i.e. $V\in C(\mathbb{R}^3)$ stisfies $V(x+z)=V(x)$ for all $z\in\mathbb{Z}^3$; and ring-shaped potential, e.g. $
 V(x)=||x|-1|^p$ for some $p>0$.
\end{abstract}

\section{Introduction}
In this paper, we study the variational problem 
\begin{equation}
I(a):=\inf\left\{ \mathcal{E}_{a}(u):u\in H^{1/2}(\mathbb{R}^{3}),\|u\|_{L^{2}}^{2}=1\right\} \label{eq:boson star energy}
\end{equation}
where $\mathcal{E}_{a}$ is the energy functional of a \emph{boson star}, given by 
\begin{equation}
\mathcal{E}_{a}(u):=\|(-\Delta+m^{2})^{1/4}u\|_{L^{2}}^{2}+\int_{\mathbb{R}^{3}}V(x)|u(x)|^{2}{\rm d}x-\frac{a}{2}\iint_{\mathbb{R}^{3}\times\mathbb{R}^{3}}\frac{|u(x)|^{2}|u(y)|^{2}}{\left|x-y\right|}{\rm d}x{\rm d}y\label{eq:boson star functional}
\end{equation}
Here the pseudo-differential operator $\sqrt{-\Delta+m^{2}}$ is defined as usual via Fourier transform, i.e. 
$$
\|(-\Delta+m^{2})^{1/4}u\|_{L^{2}}^{2}=\int_{\mathbb{R}^{3}}\sqrt{|\xi|^{2}+m^{2}}\left|\widehat{u}(\xi)\right|^{2}{\rm d}\xi.
$$
In the physical context of boson star, the parameter $m>0$ is the mass of particles. The function $V \ge 0$ stands for an external potential, which will be assumed to be trapping (i.e. $V(x)\to \infty$ as $|x|\to \infty$) or periodic (the case $V=0$ is also allowed). The coupling constant $a>0$ describes the strength of the attractive interaction.

The functional $\mathcal{E}_{a}$ effectively describes the energy per particle of a boson star. The derivation of this functional from many-body quantum theory can be found in \cite{LiYa-87,LeNaNi-14}. In this context, we can interpret $a=gN$ where $g$ is the gravitational constant and $N$ is the number of particle in the star. It is a fundamental fact that the boson star {\em collapses} when $N$ is larger than a critical number, often called the {\em  Chandrasekhar limit}. In the effective model \eqref{eq:boson star energy}, the collapse simply boils down to the fact that $I(a)=-\infty$ if $a$ is larger than a critical value $a^{*}$.   

From the simple inequality
$$ 
\sqrt{-\Delta} \le \sqrt{-\Delta+m^2} \le \sqrt{-\Delta}+m
$$
and a standard scaling argument, we can see that the critical $a^{*}$ is independent of $m$ and $V$. Indeed, it is the optimal constant in the interpolation inequality
\begin{equation}\label{eq:boson star inequality}
\|(-\Delta)^{1/4}u\|_{L^{2}}^{2}\|u\|_{L^{2}}^{2}\ge\frac{a^{*}}{2}\iint_{\mathbb{R}^{3}\times\mathbb{R}^{3}}\frac{|u(x)|^{2}|u(y)|^{2}}{\left|x-y\right|}{\rm d}x{\rm d}y,\quad\forall u\in H^{1/2}(\mathbb{R}^{3}).
\end{equation}

It is well-known (see e.g, \cite{LiYa-87,Le-07,LeLe-11}) that the
inequality \eqref{eq:boson star inequality} has an optimizer $Q\in H^{1/2}(\mathbb{R}^3)$
which can be chosen to be positive radially symmetric decreasing
and satisfy
\begin{equation}
\|(-\Delta)^{1/4}Q\|_{L^{2}}=\|Q\|_{L^{2}}=\frac{a^{*}}{2}\iint_{\mathbb{R}^{3}\times\mathbb{R}^{3}}\frac{\left|Q(x)\right|^{2}\left|Q(y)\right|^{2}}{\left|x-y\right|}{\rm d}x{\rm d}y=1.\label{eq:boson star optimizer}
\end{equation}
Moreover, $Q$ solves the nonlinear equation 
\begin{equation}
\sqrt{-\Delta}Q+Q-a^{*}(\left|\cdot\right|^{-1}\star\left|Q\right|^{2})Q=0,\quad Q\in H^{1/2}(\mathbb{R}^{3})\label{eq:massless boson star}
\end{equation}
and it satisfies the decay property (see \cite{FrLe-09})
\begin{equation} \label{eq:decay}
Q(x)\le C (1+|x|)^{-4}, \quad (\left|\cdot\right|^{-1}\star\left|Q\right|^{2})(x)\leq C(1+|x|)^{-1}.
\end{equation}

The uniqueness (up to translation and dilation) of the optimizer for \eqref{eq:boson star inequality}, as well as the uniqueness (up to translation) of the positive solution to the equation \eqref{eq:massless boson star}, is an {\em open problem} (see  \cite{LiYa-87} for related discussions). In the following, we denote by $\mathcal{G}$ the set of all positive radially symmetric decreasing functions satisfying \eqref{eq:boson star optimizer}-\eqref{eq:massless boson star}. 
\begin{equation}\label{cG}
\mathcal{G}=\{\text{all positive radial decreasing functions satisfying } \eqref{eq:boson star optimizer}-\eqref{eq:massless boson star}\}. 
\end{equation}
In the present paper, we will analyze the existence and the blow-up profile of the minimizers for the variational problem $I(a)$ in \eqref{eq:boson star energy} when $a\nearrow a^{*}$.  

Our first result is 

\begin{thm}[Existence and nonexistence of minimizers] 
	\label{thm:existence of minimizers} Assume that $m>0$ and $V$ satisfies one of
	the two conditions:
	
	\begin{itemize}
		
		\item[$(V_1)$] (Trapping potentials) $0\leq V\in L_{{\rm loc}}^{\infty}(\mathbb{R}^{3})$ and $\lim_{\left|x\right|\to\infty}V(x)=\infty$; or
		
		\item [$(V_2)$] (Periodic potentials) $0\leq V\in C(\mathbb{R}^{3})$, $V(x+z)= V(x)$ for all $z\in\mathbb{Z}^{3}$.
		
	\end{itemize}
	Then the following statements hold true:
	
	\begin{itemize}
		\item [(i)] If $a>a^{*}$, then $I(a)=-\infty$. 
		\item [(ii)] If $a=a^{*}$, then $I(a^{*})=\inf_{x\in\mathbb{R}^{3}}V(x)$, but it has no minimizer. 
		\item [(iii)] If $0\leq a<a^{*}$, then $I(a)>0$ and it has at least one minimizer.
		Moreover 
		$$
		\lim_{a\nearrow a^{*}}I(a)=I(a^{*})=\inf V.
		$$
	\end{itemize}
\end{thm}

In case $V=0$, the result is well-known (see \cite{LiYa-87,FrJoLe-07}). Otherwise, the proof of the existence for trapping potentials is based on the direct method of calculus. The existence for periodic potentials is more involved and we have to use the concentration-compactness argument \cite{Lions-84} to deal with the lack of compactness.

Note that we can restrict the minimization problem $I(a)$ to non-negative functions $u$ 
since $\mathcal{E}_{a}(u)\geq\mathcal{E}_{a}(\left|u\right|)$ for
any $u\in H^{1}(\mathbb{R}^{3})$. This follows from the fact that
$\|(-\Delta+m^{2})^{1/4}|u|\|_{L^2} \leq \|(-\Delta+m^{2})^{1/4}u\|_{L^2}$ (see \cite[Theorem 7.13]{LiLo}). In particular, the minimizer of $I(a)$, when it exists, can be chosen to be non-negative. Furthermore, when $V$ is radial increasing, one can actually restrict the minimization problem $I(a)$ to radial decreasing functions, by rearrangement inequalities (see \cite[Chapter 3]{LiLo}).

Our next results concern the behavior of the minimizer $u$ of $I(a)$ as $a\nearrow a^{*}$.
We will show that $u$ blows up and its blow-up profile is given
by a function $Q$ in $\mathcal{G}$. Of course, this blow-up process depends crucially on the local behavior
of $V$ close to its minimizers. We will consider three cases: trapping potentials growing polynomially around its minimizers, periodic potentials, and ring-shaped potentials. The choices of potentials are inspired by the recent studies on the 2D focusing Gross-Pitaevskii in \cite{GuSe-14,QiDu-17,GuZeZh-16}.

First, we are interested to the case when $V$ is a trapping potential which behaves polynomially close to its minima. We assume that $V\ge 0$, $V^{-1}(0)=\{x_i\}_{i=1}^n \subset \mathbb{R}^3$ and there exist constants $p_i>0$, $\kappa_i>0$ such that
\begin{equation} \label{bluo-trap}
 \lim_{x\to x_{i}}\frac{V(x)}{\left|x-x_{i}\right|^{p_i}}=\kappa_i,\quad \forall i=1,2,\ldots,n.
\end{equation}
 Let $p=\max\{p_i: 1\leq i \leq n\}$, $\kappa = \min\{\kappa_i: p_i=p\}$, and let 
 $$
 \mathcal{Z}:=\{x_i:p_i=p\text{ and }\kappa_i=\kappa\}
 $$
 denote the locations of the flattest global minima of $V(x)$. In this case, the blow-up profile is given in the following

\begin{thm}[Blow-up for trapping potentials with finite minimizers] 
	\label{thm:behavior-trapping} Let $V$ satisfy $(V_1)$ in Theorem \ref{thm:existence of minimizers} and the assumption \eqref{bluo-trap}. Let $u_a$ be a non-negative minimizer of $I(a)$ in \eqref{eq:boson star energy}
	for $0\leq a<a^{*}$. Then for every sequence $\left\{ a_{k}\right\} $ with
	$a_{k}\nearrow a^{*}$ as $k\to\infty$, there exist a subsequence (still denoted by $\left\{ a_{k}\right\} $ for simplicity) and an element $Q\in \mathcal{G}$ in \eqref{cG} such that the following strong convergences hold true in $H^{1/2}(\mathbb{R}^{3})$.
		\begin{itemize}
		\item If $p\le 1$, then there exists an $x_0\in\mathcal{Z}$ such that
	\begin{equation} \label{blowup-p<1}
	\lim_{k\to\infty}\left(a^{*}-a_{k}\right)^{\frac{3}{2\left(p+1\right)}}u_{a_{k}}\left(x_{0}+x\left(a^{*}-a_{k}\right)^{\frac{1}{p+1}}\right)=\lambda^{\frac{3}{2}}Q\left(\lambda x\right)
	\end{equation}
	where
	$$
		\lambda=\inf_{W\in \mathcal{G}}\left( a^{*} p \kappa \int_{\mathbb{R}^{3}}\left|x\right|^{p}\left|W(x)\right|^{2}{\rm d}x\right)^{\frac{1}{p+1}} = \left( a^{*} p \kappa \int_{\mathbb{R}^{3}}\left|x\right|^{p}\left|Q(x)\right|^{2}{\rm d}x\right)^{\frac{1}{p+1}}
		$$
	if $0<p<1$, and 	
			\begin{align*}
		\lambda & = \inf_{W\in \mathcal{G}}\left(\frac{m^{2}a^{*}}{2}\|(-\Delta)^{-1/4}W\|_{L^{2}}^{2}+a^{*}\kappa \int_{\mathbb{R}^{3}}\left|x\right|\left|W(x)\right|^{2}{\rm d}x\right)^{\frac{1}{2}}\\
		& = \left(\frac{m^{2}a^{*}}{2}\|(-\Delta)^{-1/4}Q\|_{L^{2}}^{2}+a^{*}\kappa \int_{\mathbb{R}^{3}}\left|x\right|\left|Q(x)\right|^{2}{\rm d}x\right)^{\frac{1}{2}}, \quad \text{if  }p=1.
		\end{align*}
		
	\item If $p> 1$, then there exists a sequence $\{y_k\}\subset\mathbb{R}^3$ such that
	\begin{equation} \label{blowup-p>1}
	\lim_{k\to\infty}\left(a^{*}-a_{k}\right)^{\frac{3}{4}}u_{a_{k}}\left(y_{k}+x\left(a^{*}-a_{k}\right)^{\frac{1}{2}}\right)=\lambda^{\frac{3}{2}}Q\left(\lambda x\right)
	\end{equation}
	where
	$$
		\lambda = m\sqrt{\frac{a^{*}}{2}} \inf_{W\in \mathcal{G}} \|(-\Delta)^{-1/4}W\|_{L^{2}} = m\sqrt{\frac{a^{*}}{2}} \|(-\Delta)^{-1/4}Q\|_{L^{2}}.
		$$
	 	\end{itemize}

\end{thm}

The statement of Theorem \ref{thm:behavior-trapping} looks a bit technical but the main idea is simple. Heuristically, if the minimizer $u$ collapse at a length $L\to 0$ around $x_0$, namely
$$
L^{3/2} u(x_0+Lx) \approx Q(x),
$$
then by using the formal approximation (cf. \eqref{ineq:operator})
$$
\sqrt{-\Delta+m^{2}} \approx \sqrt{-\Delta}+\frac{m^{2}}{2\sqrt{-\Delta}}
$$
and the assumption that $V(x)\approx \kappa |x-x_0|^p$ around $x_0$ we have
\begin{equation}\label{length scale}
\mathcal{E}_{a}(u) \approx \frac{1}{L}\left(1-\frac{a}{a^*} \right) + L\frac{m^2}{2} \|(-\Delta)^{-1/4}Q\|_{L^{2}}^{2} + L^p \kappa\int_{\mathbb{R}^3} |x|^{p} |Q(x)|^2{\rm d}x. 
\end{equation}
Then the result in Theorem \ref{thm:behavior-trapping} essentially follows by obtimizing over $L>0$ on the right side of \eqref{length scale} (see estimations \eqref{lim:I(a_{k})}-\eqref{lim:V} in the proof for more details). In this way, we also obtain the asymptotic behavior of the ground state energy
\begin{equation}\label{lim: I(a_k)/epsilon_k}
\lim_{a\nearrow a^{*} }\frac{I(a)}{(a^{*}-a)^{\frac{q}{q+1}}}= \frac{q+1}{q} \cdot \frac{\lambda}{a^{*}}, \quad \text{with}\quad q=\min\{p,1\}.
\end{equation}

In the case $V=0$, the blow-up profile of minimizers of $I(a)$ has been studied in \cite{GuZe-17,Ng-17}. Indeed, this case can be interpreted as a special case of \eqref{blowup-p>1} with $p=\infty$. The convergence \eqref{blowup-p<1} for power $0<p<1$ has been also proved in a recent, independent work of Yang-Yang \cite{YaYa-17}. Our proof is somewhat simpler than that in \cite{YaYa-17} and we obtain the convergence in $H^{1/2}(\mathbb{R}^{3})$ instead of $L^{2}(\mathbb{R}^{3})$. Moreover, we are able to deal with the full range $p>0$, which is interesting since the blow-up speed changes when $p\ge 1$. 

To analyze the detailed behavior of minimizers of $I(a)$, delicate estimates on kinetic energy and potential energy are required.
When $p>1$ we lose information about the sequence $y_k$ in \eqref{blowup-p>1}, since $V$ has no impact to the leading order of $I(a)$. However, if $V$ is strictly radial increasing, for example $V(x)=|x|^p$, then the minimizers must be radial decreasing and hence we can choose $y_k=0$.

Next, we come to the cases of periodic and ring-shaped potentials. 
\begin{thm}[Blow-up for periodic potentials]
	\label{thm:behavior-periodic} Let $V$ satisfy $(V_2)$ in Theorem \ref{thm:existence of minimizers}. Assume further that $V^{-1}(0)=x_0+\mathbb{Z}^3$ for some $x_0\in [0,1]^3$ and there exist $p>0$, $\kappa>0$ such that
	$$
	\lim_{x\to x_0}\frac{V(x)}{\left|x-x_0\right|^{p}}=\kappa>0.
	$$
	Let $u_a$ be a non-negative minimizer of $I(a)$ in \eqref{eq:boson star energy}
	for $0\leq a<a^{*}$. Then for every sequence $\left\{ a_{k}\right\} $ with
	$a_{k}\nearrow a^{*}$ as $k\to\infty$, there exist a subsequence of $\left\{ a_{k}\right\}$ (still denoted by $\left\{ a_{k}\right\}$ for simplicity) and an element $Q\in \mathcal{G}$ in \eqref{cG}  such that the following strong convergences hold true in $H^{1/2}(\mathbb{R}^{3})$.
		\begin{itemize}
		\item If $p\le 1$, then there exists a sequence $\{z_k\}\subset\mathbb{Z}^3$ such that
	$$
	\lim_{k\to\infty}\left(a^{*}-a_{k}\right)^{\frac{3}{2\left(p+1\right)}}u_{a_{k}}\left(x_{0}+z_k+x\left(a^{*}-a_{k}\right)^{\frac{1}{p+1}}\right)=\lambda^{\frac{3}{2}}Q\left(\lambda x\right)
	$$
where $\lambda$ is defined as in \eqref{blowup-p<1}.	
	\item If $p> 1$, then there exists a sequence $\{y_k\}\subset\mathbb{R}^3$ such that
	$$
	\lim_{k\to\infty}\left(a^{*}-a_{k}\right)^{\frac{3}{4}}u_{a_{k}}\left(y_{k}+x\left(a^{*}-a_{k}\right)^{\frac{1}{2}}\right)=\lambda^{\frac{3}{2}}Q\left(\lambda x\right)
	$$
	where $\lambda$ is defined as in \eqref{blowup-p>1}.
	 	\end{itemize}
	 	
\end{thm}

\begin{thm}[Blow-up for trapping ring-shaped potentials] 
	\label{thm:behavior-ring-shaped} Let $V(x)=\left|\left|x\right|-1\right|^{p}$
	for some $p>0$. Let $u_a$ be a non-negative minimizer of $I(a)$ in \eqref{eq:boson star energy}
	for $0\leq a<a^{*}$. Then for every sequence $\left\{ a_{k}\right\} $ with
	$a_{k}\nearrow a^{*}$ as $k\to\infty$, there exist a subsequence of $\left\{ a_{k}\right\}$ (still denoted by $\left\{ a_{k}\right\}$ for simplicity) and an element $Q\in \mathcal{G}$ in \eqref{cG} such that the following strong convergences hold true in $H^{1/2}(\mathbb{R}^{3})$.
	\begin{itemize}

		\item If $0<p\le 1$, then there exists a sequence $\{x_{k}\}\subset\mathbb{R}^3$, $|x_k|\to 1$ such that
	\begin{equation} \label{blowup-ring}
	\lim_{k\to\infty}\left(a^{*}-a_{k}\right)^{\frac{3}{2\left(p+1\right)}}u_{a_{k}}\left(x_{k}+x\left(a^{*}-a_{k}\right)^{\frac{1}{p+1}}\right)=\lambda^{\frac{3}{2}}Q\left(\lambda x\right)
	\end{equation}
	where $$
		\lambda=\inf_{W\in \mathcal{G}}\left(pa^{*}\int_{\mathbb{R}^{3}}\left|x_{0}\cdot x\right|^{p}\left|W(x)\right|^{2}{\rm d}x\right)^{\frac{1}{p+1}} = \left(pa^{*}\int_{\mathbb{R}^{3}}\left|x_{0}\cdot x\right|^{p}\left|Q(x)\right|^{2}{\rm d}x\right)^{\frac{1}{p+1}}
		$$
		if $0<p<1$ and
		\begin{align*}
		\lambda & = \inf_{W\in \mathcal{G}}\left(\frac{m^{2}a^{*}}{2}\|(-\Delta)^{-1/4}W\|_{L^{2}}^{2}+a^{*}\int_{\mathbb{R}^{3}}\left|x_{0}\cdot x\right|\left|W(x)\right|^{2}{\rm d}x\right)^{\frac{1}{2}}\\
		& = \left(\frac{m^{2}a^{*}}{2}\|(-\Delta)^{-1/4}Q\|_{L^{2}}^{2}+a^{*}\int_{\mathbb{R}^{3}}\left|x_{0}\cdot x\right|\left|Q(x)\right|^{2}{\rm d}x\right)^{\frac{1}{2}} \quad \text{if  }p=1.
		\end{align*}
		
		\item If $p>1$, then there exists a sequence $\{y_k\}\subset\mathbb{R}^3$ such that
	$$
	\lim_{k\to\infty}\left(a^{*}-a_{k}\right)^{\frac{3}{4}}u_{a_{k}}\left(y_{k}+x\left(a^{*}-a_{k}\right)^{\frac{1}{2}}\right)=\lambda^{\frac{3}{2}}Q\left(\lambda x\right)
$$
where $\lambda$ is defined as in \eqref{blowup-p>1}.
	\end{itemize}
\end{thm}

These cases of periodic and ring-shaped potentials are interesting because the potentials have infinitely many minimizers. Consequently, the variational problem $I(a)$ in \eqref{eq:boson star energy} might have infinitely many minimizers. Moreover, in the case of ring-shaped potentials, although the energy functional $\mathcal{E}_a(u)$ is invariant under rotations, its minimizers might be not radially symmetric when $a \nearrow a^*$. Indeed, \eqref{blowup-ring} shows that the minimizer $u_{a_k}$ is not radially symmetric because it concentrates around a point in the unit sphere. This implies that {\em symmetry breaking} occurs.

\medskip
{\bf Organization of the paper.} In Section \ref{sec:theorem of existence} we prove the existence and non-existence of minimizers of $I(a)$ in \eqref{eq:boson star energy}. In Section \ref{sec:behavior}, we will give the proof of Theorem \ref{thm:behavior-trapping}, \ref{thm:behavior-periodic} and \ref{thm:behavior-ring-shaped} which give the blow-up profiles of minimizers of $I(a)$.

\medskip

\section{\label{sec:theorem of existence}Proof of existence and non-existence of minimizer}

In this section, we prove Theorem \ref{thm:existence of minimizers}. We assume without loss of generality that $\inf_{x\in\mathbb{R}^{3}}V(x)=0$.

For $0\leq a<a^{*}$, let $\left\{ u_{k}\right\}$ be a minimizing sequence for $I(a)$, i.e.,
$$
\lim_{k\to\infty}\mathcal{E}_{a}(u_{k})=I(a), \quad \text{ with } \left\{ u_{k}\right\} \in H^{1/2}(\mathbb{R}^{3}) \text{ and } \|u_{k}\|_{L^{2}}^{2}=1 \text{ for all } k\geq0.
$$
We observe from \eqref{eq:boson star inequality}
that 
\begin{equation}\label{ineq:boundeness of I(a)}
\mathcal{E}_{a}(u_{k})\geq\left(1-\frac{a}{a^{*}}\right)\|(-\Delta)^{1/4}u_{k}\|_{L^{2}}^{2}+\int_{\mathbb{R}^{3}}V(x)\left|u_{k}(x)\right|^{2}{\rm d}x.
\end{equation}
Thus $I(a)>-\infty$ and $\left\{ u_{k}\right\}$ is a bounded sequence in $H^{1/2}(\mathbb{R}^{3})$. Hence, extracting a subsequence if necessary, we assume that $u_{k}\rightharpoonup u$
weakly in $H^{1/2}(\mathbb{R}^{3})$ and $u_{k}\to u$ a.e. in $\mathbb{R}^3$. Moreover we have $u_{k}\to u$ strongly in $L^{r}_{{\rm loc}}(\mathbb{R}^{3})$ for $2\leq r<3$, thanks to a Rellich-type theorem for $H^{1/2}(\mathbb{R}^{3})$ (see, e.g, \cite[Theorem 8.6]{LiLo}).

\subsection{\label{subsec:Trapping-potential}Existence of minimizers in the case of a trapping potential}

We first consider $V$ be trapping potential, that means $V$ satisfies $\left(V_{1}\right)$. 
\begin{lem}
	\label{lem:compact embedding}Suppose $V\in L_{{\rm loc}}^{\infty}\left(\mathbb{R}^{3}\right)$
	with $V(x)\to\infty$ as $\left|x\right|\to\infty$. If $\left\{ u_{k}\right\} $
	is a bounded sequence in $H^{1/2}(\mathbb{R}^{3})$ and satisfies that $\int_{\mathbb{R}^{3}}V(x)\left|u_{k}(x)\right|^{2}{\rm d}x\leq C$, then $u_{k}\to u$ strongly in $L^{r}(\mathbb{R}^{3})$ for $2\leq r<3$. 
\end{lem}
\begin{proof}
	A similar proof to this Lemma can be found in \cite{AdRo-03}.
\end{proof}

From \eqref{ineq:boundeness of I(a)}, we see that $\int_{\mathbb{R}^{3}}V(x)\left|u_{k}(x)\right|^{2}{\rm d}x$
is uniformly bounded in $k$. By Lemma \ref{lem:compact embedding} and extracting a subsequence if necessary, we have $u_{k}\to u$ strongly in $L^{r}(\mathbb{R}^{3})$ for $2\leq r<3$.
We conclude that $\|u\|_{L^{2}}^{2}=1$. By the Hardy-Littlewood-Sobolev inequality (see, e.g, \cite[Theorem 4.3]{LiLo}) we have 
$$
\lim_{k\to\infty}\iint_{\mathbb{R}^{3}\times\mathbb{R}^{3}}\frac{\left|u_{k}(x)\right|^{2}\left|u_{k}(y)\right|^{2}}{\left|x-y\right|}{\rm d}x{\rm d}y = \iint_{\mathbb{R}^{3}\times\mathbb{R}^{3}}\frac{|u(x)|^{2}|u(y)|^{2}}{\left|x-y\right|}{\rm d}x{\rm d}y.
$$
Thus, by weak lower semicontinuity we have
$$
I(a)=\lim_{k\to\infty}\mathcal{E}_{a}\left(u_{k}\right)\geq\mathcal{E}_{a}(u)\geq I(a)
$$
which show that $\mathcal{E}_{a}(u)=I(a)$. This implies the existence
of minimizers of \eqref{eq:boson star energy} for any $0\leq a<a^{*}$.
At the same time, by $\|u\|_{L^{2}}^{2}=1$ and \eqref{ineq:boundeness of I(a)},
we also have that $I(a)>0$.

To prove that there is no minimizer for \eqref{eq:boson star energy}
as soon as $a\geq a^{*}$, we proceed as follow. Let $Q\in \mathcal{G}$ in \eqref{cG}. Choose $0\leq\chi\leq1$ be a fixed smooth funtion in $\mathbb{R}^{3}$
such that $\chi\equiv1$ for $\left|x\right|<1$ and $\chi\equiv0$
for $\left|x\right|\geq2$. For $R,\tau>0$ and $x_{0}\in\mathbb{R}^{3}$
satisfies $V(x_{0})=\inf_{x\in\mathbb{R}^{3}}V(x)=0$, we define the
functions $\chi_{R}(x)=\chi\left(\frac{x}{R}\right)$, $\zeta_{R}(x)=\sqrt{1-\chi_{R}(x)^{2}}$
and the trial state
\begin{equation}\label{eq:trial state trapping-periodic}
u_{R,\tau}(x)=A_{R,\tau}\tau^{3/2}\chi_{R}\left(x-x_{0}\right)Q\left(\tau\left(x-x_{0}\right)\right)
\end{equation}
where $A_{R,\tau}$ is chosen so that $\|u_{R,\tau}\|_{L^{2}}^{2}=1$
as $\tau\to\infty$. In fact, by algebraic decay rate of $Q$ in \eqref{eq:decay} we have 
$$
\frac{1}{A_{R,\tau}^{2}}=\int_{\mathbb{R}^{3}}\chi_{R\tau}(x)^{2}\left|Q(x)\right|^{2}{\rm d}x=1+o(1)_{R\tau\to\infty},
$$
where $o(1)_{R\tau\to\infty}$ means a quantity that converges to $0$ as $R\tau\to\infty$.
In the following we could set $R=1$, for instance.

We have the operator inequality
\begin{equation}\label{ineq:operator}
\frac{m^{2}}{2\sqrt{-\Delta+m^{2}}}\leq \sqrt{-\Delta+m^{2}}-\sqrt{-\Delta}=\frac{m^{2}}{\sqrt{-\Delta+m^{2}}+\sqrt{-\Delta}}\leq\frac{m^{2}}{2\sqrt{-\Delta}},
\end{equation}
with notice that
$$
\|(-\Delta)^{-1/4}u_{R,\tau}\|_{L^{2}}^{2}=\frac{A_{R,\tau}^{2}}{\tau}\|(-\Delta)^{-1/4}\chi_{R\tau}Q\|_{L^{2}}^{2}\leq\frac{1}{\tau}\|(-\Delta)^{-1/4}Q\|_{L^{2}}^{2}+o(1)_{R\tau\to\infty}.
$$

By the \emph{IMS-type localization formulas} (see e.g, \cite{LeLe-10,LiYa-88}) we have 
$$
\|(-\Delta)^{1/4}\chi_{R\tau}Q\|_{L^{2}}^{2}+\|(-\Delta)^{1/4}\zeta_{R\tau}Q\|_{L^{2}}^{2}\leq\|(-\Delta)^{1/4}Q\|_{L^{2}}^{2}+\frac{C}{(R\tau)^{2}},
$$
which implies that 
$$
\|(-\Delta)^{1/4}u_{R,\tau}\|_{L^{2}}^{2}=A_{R,\tau}^{2}\tau\|(-\Delta)^{1/4}\chi_{R\tau}Q\|_{L^{2}}^{2}\leq\tau+o(1)_{R\tau\to\infty}.
$$

On the other hand, we have 
\begin{align*}
& \iint_{\mathbb{R}^{3}\times\mathbb{R}^{3}}\frac{\chi_{R\tau}(x)^{2}\chi_{R\tau}(y)^{2}\left|Q(x)\right|^{2}\left|Q(y)\right|^{2}}{\left|x-y\right|}{\rm d}x{\rm d}y\\
& \geq\iint_{\mathbb{R}^{3}\times\mathbb{R}^{3}}\frac{\left|Q(x)\right|^{2}\left|Q(y)\right|^{2}}{\left|x-y\right|}{\rm d}x{\rm d}y-2\iint_{\mathbb{R}^{3}\times\mathbb{R}^{3}}\frac{\zeta_{R\tau}(x)^{2}\left|Q(x)\right|^{2}\left|Q(y)\right|^{2}}{\left|x-y\right|}{\rm d}x{\rm d}y\\
& =\frac{2}{a^{*}}+o(1)_{R\tau\to\infty},
\end{align*}
because of decay rate of $Q(x)$ and $(\left|\cdot\right|^{-1}\star\left|Q\right|^{2})(x)$ in \eqref{eq:decay}.
This implies that 
\begin{align*}
& \iint_{\mathbb{R}^{3}\times\mathbb{R}^{3}}\frac{\left|u_{R,\tau}(x)\right|^{2}\left|u_{R,\tau}(y)\right|^{2}}{\left|x-y\right|}{\rm d}x{\rm d}y\\
& =A_{R,\tau}^{4}\tau\iint_{\mathbb{R}^{3}\times\mathbb{R}^{3}}\frac{\chi_{R\tau}(x)^{2}\chi_{R\tau}(y)^{2}\left|Q(x)\right|^{2}\left|Q(y)\right|^{2}}{\left|x-y\right|}{\rm d}x{\rm d}y\geq\frac{2\tau}{a^{*}}+o(1)_{R\tau\to\infty}.
\end{align*}

Moreover, since $V(x)\chi_{R\tau}(x-x_{0})^{2}$ is bounded and has
compact support, we have 
\begin{equation}\label{eq:thm1.1-1}
\int_{\mathbb{R}^{3}}V(x)\left|u_{R,\tau}(x)\right|^{2}{\rm d}x=V(x_{0})+o(1)_{R\tau\to\infty}=o(1)_{R\tau\to\infty}.
\end{equation}

Thus, we have proved that 
\begin{equation}
\mathcal{E}_{a}\left(u_{R,\tau}\right)\leq\tau\left(1-\frac{a}{a^{*}}\right)+\frac{m^{2}}{2\tau}\|(-\Delta)^{-1/4}Q\|_{L^{2}}^{2}+o(1)_{R\tau\to\infty}.\label{ineq:trial state-1}
\end{equation}

Now, we are able to prove the non-existence of minimizer of $I(a)$
when $a\geq a^{*}$. For $a=a^{*}$, it follows from \eqref{ineq:boundeness of I(a)} and \eqref{ineq:trial state-1} that 
$$
0\leq I(a^{*})\leq\lim_{\tau\to\infty}\mathcal{E}_{a^{*}}\left(u_{R,\tau}\right)=0.
$$
Thus $I(a^{*})=0$. To prove that there does not exist a minimizer for $I(a)$ with $a=a^{*}$, we argue by contradiction as follows. We suppose that $u\in H^{1/2}(\mathbb{R}^{3})$ is a minimizer for $I(a)$ with $a=a^{*}$. It follows from the strict inequality $\|(-\Delta+m^{2})^{1/4}u\|_{L^{2}}>\|(-\Delta)^{1/4}u\|_{L^{2}}$ that
$$
0=\mathcal{E}_{a^{*}}(u)|_{m>0}>\mathcal{E}_{a^{*}}(u)|_{m=0}\geq0,
$$
which is a contradiction. Hence no minimizer exists for $I(a)$ if $a=a^{*}$.

For $a>a^{*}$, it follows from
\eqref{ineq:trial state-1} that 
$$
I(a)\leq\lim_{\tau\to\infty}\mathcal{E}_{a}\left(u_{R,\tau}\right)=-\infty.
$$
This implies that $I(a)$ is unbounded from below for any $a>a^{*}$, and hence the non-existence of minimizers is therefore proved.

It remains to prove $\lim_{a\to a^{*}}I(a)=0$. This follows easily from \eqref{ineq:trial state-1}, by first taking $a\to a^{*}$, followed
by $\tau\to\infty$.

\subsection{\label{subsec:Periodic-potential}Existence of minimizers in the case of a periodic potential}

We now consider $V$ be periodic potential, that means $V$ satisfies
$\left(V_{2}\right)$. This case is harder than the previous one, and we will need to use the concentration-compactness argument \cite{Lions-84}.

The non-existence of minimizer of $I(a)$ follows from the same argument of non-existence part in Section \ref{subsec:Trapping-potential}. To deal with the existence of minimizer of $I(a)$ when $0\leq a<a^{*}$, we will need some priliminary results. 

Let
$$
I_{c}(a) =\inf\left\{ \mathcal{E}_{a}(u):u\in H^{1/2}(\mathbb{R}^{3}),\|u\|_{L^{2}}^{2}=c\right\}, 
$$
we note that $I(a)=I_{1}(a)$.

\begin{lem}
	\label{lem:continuous of ennegy}$I_{c}(a)$ is uniformly continuous with respect to $c\in (0,1)$. 
\end{lem}
\begin{proof}
	By a simple scaling we have $I_{c}(a)=cI(ac)$. We remark that $$I(ac)=\inf\{\text{linear functions in }c\}$$ has to be a concave function with respect to $c$. Therefore it follows that $I_{c}(a)$ is a concave function with respect to $c$. Thus $I_{c}(a)$ is uniformly continuous with respect to $c\in (0,1)$.
\end{proof}

\begin{lem}
	For any $0<a<a^{*}$ and $0<\lambda<1$, we have 
	\begin{equation}
	I(a)<I_{\lambda}(a)+I_{1-\lambda}(a).\label{eq:dichotomy}
	\end{equation} 
\end{lem}
\begin{proof}
	We show that for any $1<\theta\leq\frac{1}{\delta}$
	and $\delta\in\left(0,1\right)$ then $I_{\theta\delta}(a)\leq\theta I_{\delta}(a)$.
	Indeed, 
	\begin{align*}
	I_{\theta\delta}(a) & =\inf_{\|u\|_{L^{2}}^{2}=\theta\delta} \mathcal{E}_{a}(u) =\inf_{\|u\|_{L^{2}}^{2}=\delta} \mathcal{E}_{a}\left(\theta^{\frac{1}{2}}u\right) \\
	& =\inf_{\|u\|_{L^{2}}^{2}=\delta}\left\{ \theta\mathcal{E}_{a}(u)+\frac{\theta-\theta^{2}}{2}\iint_{\mathbb{R}^{3}\times\mathbb{R}^{3}}\frac{|u(x)|^{2}|u(y)|^{2}}{\left|x-y\right|}{\rm d}x{\rm d}y\right\} \\
	& <\theta I_{\delta}(a).
	\end{align*}
	
	Thus, by an argument presented in \cite[Lemma II.1]{Lions-84}, this inequality leads to the strict subadditivity condition \eqref{eq:dichotomy}.
\end{proof}

\begin{lem}
	\label{lem:dichotomy1}Let $2<r<\infty$ and let $u_{k}$, $u$
	be functions in $\mathbb{R}^{3}$ such that $u_{k}(x)\to u(x)$ a.e.
	in $\mathbb{R}^{3}$, $\sup_{k}\|u_{k}\|_{L^{r}}<\infty$ and 
	$$
	\sup_{k}\iint_{\mathbb{R}^{3}\times\mathbb{R}^{3}}\frac{\left|u_{k}(x)\right|^{2}\left|u_{k}(y)\right|^{2}}{\left|x-y\right|}{\rm d}x{\rm d}y<\infty,
	$$
	then 
	\begin{align*}
	\iint_{\mathbb{R}^{3}\times\mathbb{R}^{3}}\frac{\left|u_{k}(x)\right|^{2}\left|u_{k}(y)\right|^{2}}{\left|x-y\right|}{\rm d}x{\rm d}y 
	&
	=\iint_{\mathbb{R}^{3}\times\mathbb{R}^{3}}\frac{\left|u_{k}(x)-u(x)\right|^{2}\left|u_{k}(y)-u(y)\right|^{2}}{\left|x-y\right|}{\rm d}x{\rm d}y\\
	& \quad +\iint_{\mathbb{R}^{3}\times\mathbb{R}^{3}}\frac{|u(x)|^{2}|u(y)|^{2}}{\left|x-y\right|}{\rm d}x{\rm d}y+o(1).
	\end{align*}
\end{lem}
\begin{proof}
	The proof of this Lemma can be found in \cite[Lemma 2.2]{BeFrVi-14}. 
\end{proof} 

Now, we claim that there exists a sequence $\{y_{k}\}\subset\mathbb{R}^3$ and a postitive constant $R_{0}$ such that 
\begin{equation}\label{eq:vanishing0}
\liminf_{k\to\infty}\int_{B(y_{k},R_{0})}\left|u_{k}(x)\right|^{2}{\rm d}x>0.
\end{equation}
On the contrary, we assume that for any $R>0$ there exists a subsequence of $\{a_{k}\}$, still denoted by $\{a_{k}\}$, such that
\begin{equation}\label{eq:vanishing0b}
\lim_{k\to\infty}\sup_{y\in\mathbb{R}^3}\int_{B(y,R)}\left|u_{k}(x)\right|^{2}{\rm d}x=0.
\end{equation}
It follows from \cite[Lemma 9]{LeLe-11}
that $\lim_{k\to\infty}\|u_{k}\|_{L^{q}}=0$ for any $2\leq q<3$. By the Hardy-Littlewood-Sobolev inequality we have 
\begin{equation}
\iint_{\mathbb{R}^{3}\times\mathbb{R}^{3}}\frac{\left|u_{k}(x)\right|^{2}\left|u_{k}(y)\right|^{2}}{\left|x-y\right|}{\rm d}x{\rm d}y\to0.\label{eq:vanishing-periodict}
\end{equation}
It follows from \eqref{eq:vanishing-periodict} that 
\begin{align*}
0<m & \leq\|(-\Delta+m^{2})^{1/4}u_{k}\|_{L^{2}}^{2}+\int_{\mathbb{R}^{3}}V(x)\left|u_{k}\right|^{2}{\rm d}x \\
&=\mathcal{E}_{a}(u_{k})+\frac{a}{2}\iint_{\mathbb{R}^{3}\times\mathbb{R}^{3}}\frac{\left|u_{k}(x)\right|^{2}\left|u_{k}(y)\right|^{2}}{\left|x-y\right|}{\rm d}x{\rm d}y\to I(a).
\end{align*}
This is impossible since $I(a)\leq C(a^{*}-a)^{\frac{q}{q+1}}$ which close to $0$ as $a$ closes to $a^{*}$. We refer to \eqref{ineq: upper bound I(a_{k}) periodic bad} for the proof of this upper bound of $I(a)$. Here $C$ is positive constant, and $q=\min\{p,1\}$.

Thus, \eqref{eq:vanishing0b} does not occurs, and hence \eqref{eq:vanishing0} holds true. For all $k\in\mathbb{N}$, we can choose $\{z_{k}\}\subset\mathbb{Z}^3$ such that $y_{k}-z_{k}\in [0,1]^3$, then we have $|y_{k}-z_{k}|\leq 2$ for any $k$. Setting $\tilde{u}_{k}(x)=u_{k}(x+z_{k})$, we then have $\{\tilde{u}_{k}\}$ is bounded in $H^{1/2}(\mathbb{R}^3)$ and $\|\tilde{u}_{k}\|_{L^{2}}=\|u_{k}\|_{L^{2}}=1$. Hence, extracting a subsequence if necessary, we assume that $\tilde{u}_{k}\rightharpoonup \tilde{u}$
weakly in $H^{1/2}(\mathbb{R}^{3})$ and $\tilde{u}_{k}\to \tilde{u}$ a.e. in $\mathbb{R}^3$. Moreover we have $\tilde{u}_{k}\to \tilde{u}$ strongly in $L^{r}_{{\rm loc}}(\mathbb{R}^{3})$ for $2\leq r<3$. Thus we deduce from \eqref{eq:vanishing0} that
\begin{align*}
& \int_{B(0,R_{0}+2)}\left|\tilde{u}(x)\right|^{2}{\rm d}x = \lim_{k\to\infty}\int_{B(0,R_{0}+2)}\left|\tilde{u}_{k}(x)\right|^{2}{\rm d}x\\
= & \lim_{k\to\infty}\int_{B(z_{k},R_{0}+2)}\left|u_{k}(x)\right|^{2}{\rm d}x \geq  \lim_{k\to\infty}\int_{B(y_{k},R_{0})}\left|u_{k}(x)\right|^{2}{\rm d}x>0.
\end{align*}
This implies that $\tilde{u}\not \equiv 0$, and hence $0<\|\tilde{u}\|_{L^{2}}^{2}:=\alpha\leq1$. We suppose that $\alpha<1$. It follows from the Brezis-Lieb refinement of Fatou's Lemma (see, e.g, \cite{Br-Li-83,LiLo}) that 
$$
\lim_{k\to\infty}\|\tilde{u}_{k}-\tilde{u}\|_{L^{2}}^{2}=\lim_{k\to\infty}\|\tilde{u}_{k}\|_{L^{2}}^{2}-\|\tilde{u}\|_{L^{2}}^{2}=1-\alpha.
$$
Since $\tilde{u}_{k}\rightharpoonup \tilde{u}$ weakly in $H^{1/2}(\mathbb{R}^{3})$
we have $\langle \tilde{u}_{k}-\tilde{u},\tilde{u}\rangle_{H^{1/2}}\to0$. Furthermore, by the periodicity of $V$ we have $\mathcal{E}_{a}(\tilde{u}_{k})=\mathcal{E}_{a}(u_{k})$. Thus we deduce from Lemma \ref{lem:continuous of ennegy} and \ref{lem:dichotomy1} that
$$
I(a)=\lim_{k\to\infty}\mathcal{E}_{a}(\tilde{u}_{k})\geq\liminf_{k\to\infty}\mathcal{E}_{a}(\tilde{u}_{k}-\tilde{u})+\mathcal{E}_{a}(\tilde{u})\geq I_{1-\alpha}(a)+I_{\alpha}(a),
$$
which contradicts \eqref{eq:dichotomy}. Hence $\|\tilde{u}\|_{L^{2}}^{2}=1$. By the same argument of proof of existence part in Section \ref{subsec:Trapping-potential}, we can show that $\tilde{u}$ is a minimizer of $I(a)$.

\section{\label{sec:behavior} Behavior of Minimizers}

We have proved in Theorem \ref{thm:existence of minimizers} that $I(a)$ has minimizers when $0\leq a<a^*$. In this section, we prove the blow-up profiles of minimizers of $I(a)$, which are presented in Theorem \ref{thm:behavior-trapping}, \ref{thm:behavior-periodic} and \ref{thm:behavior-ring-shaped}.

Let $a_{k}\nearrow a^{*}$ as $k\to\infty$ and let $u_{k}:=u_{a_{k}}$ be a non-negative minimizer for $I(a_{k})$. Such $u_{k}$ solves the Euler-Lagrange equation 
$$
\sqrt{-\Delta+m^{2}}u_{k}(x)+V(x)u_{k}(x)=\mu_{k}u_{k}(x)+a_{k}\left(\left|\cdot\right|^{-1}\star\left|u_{k}\right|^{2}\right)(x)u_{k}(x),
$$
where $\mu_{k}\in\mathbb{R}$ be a suitable Lagrange multiplier. In fact,
\begin{equation}\label{eq:Lagrange mutiplier}
\mu_{k}=I(a_{k})-\frac{a_{k}}{2}\iint_{\mathbb{R}^{3}\times\mathbb{R}^{3}}\frac{\left|u_{k}(x)\right|^{2}\left|u_{k}(y)\right|^{2}}{\left|x-y\right|}{\rm d}x{\rm d}y.
\end{equation}

In the following lemma, we will show that the compactness follows when the Euler-Lagrange multiplier stays away from 0. This is a classical idea, which goes back to Lions \cite{Lions-87}.

\begin{lem}
	\label{lem:positive ground state solution} For any sequence $\{z_{k}\}\subset\mathbb{R}^{3}$, and $\epsilon_{k}>0$ such that $\epsilon_{k}\to0$ as $k\to\infty$, let $w_{k}(x):=\epsilon_{k}^{3/2}u_{k}(\epsilon_{k}x+z_{k})$ be $L^{2}$--normalized of $u_{k}$. Assume that $w_{k}$ is bounded in $H^{1/2}(\mathbb{R}^{3})$ and $\epsilon_{k}\mu_{k}\to-\lambda<0$ as $k\to\infty$. Then there exists
	a non-negative $w\in H^{1/2}(\mathbb{R}^{3})$ such that $w_{k}\to w$
	strongly in $H^{1/2}(\mathbb{R}^{3})$. Moreover if $w>0$ then, up to translations, we
	have 
	$$
	w(x)=\lambda^{3/2}Q(\lambda x)
	$$
	where $Q\in \mathcal{G}$ in \eqref{cG}. 
\end{lem}
\begin{proof}
	We see that $w_{k}$ is a non-negative solution of 
	\begin{equation}
	\sqrt{-\Delta+m^{2}\epsilon_{k}^{2}}w_{k}(x)+\epsilon_{k}V(\epsilon_{k}x+z_{k})w_{k}(x)=\epsilon_{k}\mu_{k}w_{k}(x)+a_{k}\left(|\cdot\right|^{-1}\star\left|w_{k}\right|^{2})(x)w_{k}(x).\label{eq:E-L equation for w_a-1}
	\end{equation}
	
	Since $\{w_{k}\}$ is bounded in $H^{1/2}(\mathbb{R}^{3})$, extracting a subsequence if necessary, we assume that $w_{k}\rightharpoonup w$ weakly in $H^{1/2}(\mathbb{R}^{3})$ and $w_{k}\to w$ a.e. in $\mathbb{R}^3$. Passing \eqref{eq:E-L equation for w_a-1} to weak limit, we have $w$ be non-negative solution to
	$$
	\sqrt{-\Delta}w(x)=-\lambda w(x)+a^{*}\left(\left|\cdot\right|^{-1}\star\left|w\right|^{2}\right)(x)w(x).
	$$
	By a simple scaling we see that 
	$$
	Q(x)=\lambda^{-3/2}w(\lambda^{-1}x)
	$$
	is a non-negative solution of \eqref{eq:massless boson star}. It is well-known that the kernel of the resolvent $((-\Delta)^{1/2}+1)^{-1}$ is positive everywhere in $\mathbb{R}^{3}$ (see \cite[eq. (A.11)]{Le-07}). Hence, from \eqref{eq:massless boson star} we obtain that either $Q\equiv0$ or $Q>0$ in $\mathbb{R}^{3}$. 
	
	In the latter case, we prove that $Q\in \mathcal{G}$. First, since $Q$ is a positive solution of \eqref{eq:massless boson star}, it follows from \cite[Theorem 1.1]{FrLe-09} that $Q$ is positive radially symmetric decreasing up to translation. It remains to prove that $Q$ satisfies \eqref{eq:boson star optimizer}. Indeed, by the same reason that $Q$ is a solution of \eqref{eq:massless boson star}, we have 
	\begin{equation}
	\|(-\Delta)^{1/4}Q\|_{L^{2}}^{2}+\|Q\|_{L^{2}}^{2}=a^{*}\iint_{\mathbb{R}^{3}\times\mathbb{R}^{3}}\frac{\left|Q(x)\right|^{2}\left|Q(y)\right|^{2}}{\left|x-y\right|}{\rm d}x{\rm d}y.\label{ineq:ground state contrary 1}
	\end{equation}
	Thus, by \eqref{eq:boson star inequality} and the Cauchy-Schwarz inequality
	we have 
	\begin{equation}
	1\leq\frac{2\|(-\Delta)^{1/4}Q\|_{L^{2}}^{2}\|Q\|_{L^{2}}^{2}}{a^{*}\iint_{\mathbb{R}^{3}\times\mathbb{R}^{3}}\frac{\left|Q(x)\right|^{2}\left|Q(y)\right|^{2}}{\left|x-y\right|}{\rm d}x{\rm d}y}\leq\frac{a^{*}}{2}\iint_{\mathbb{R}^{3}\times\mathbb{R}^{3}}\frac{\left|Q(x)\right|^{2}\left|Q(y)\right|^{2}}{\left|x-y\right|}{\rm d}x{\rm d}y.\label{ineq:ground state contrary 2}
	\end{equation}
	
	Let $S$ be in $\mathcal{G}$, then  we have
	$$
	1=\|S\|_{L^{2}}^{2}=\|(-\Delta)^{1/4}S\|_{L^{2}}^{2}=\frac{a^{*}}{2}\iint_{\mathbb{R}^{3}\times\mathbb{R}^{3}}\frac{\left|S(x)\right|^{2}\left|S(y)\right|^{2}}{\left|x-y\right|}{\rm d}x{\rm d}y.
	$$
	We define $s(x)=\left(\frac{\lambda}{\epsilon_{k}}\right)^{3/2}S\left(\frac{\lambda}{\epsilon_{k}}x\right)$. Then we have $\|s\|_{L^{2}}^{2}=\|S\|_{L^{2}}^{2}=1$ and
	\begin{align*}
	\epsilon_{k}\mathcal{E}_{a_{k}}(s) & =\lambda\left(\|\left(-\Delta+m^{2}\epsilon_{k}^{2}\lambda^{-2}\right)^{1/4}S\|_{L^{2}}^{2}-\frac{a_{k}}{2}\iint_{\mathbb{R}^{3}\times\mathbb{R}^{3}}\frac{\left|S(x)\right|^{2}\left|S(y)\right|^{2}}{\left|x-y\right|}{\rm d}x{\rm d}y\right)\\
	& \quad+\epsilon_{k}\int_{\mathbb{R}^{3}}V\left(\frac{\epsilon_{k}}{\lambda}x\right)\left|S(x)\right|^{2}{\rm d}x\\
	& \leq \lambda\left(1-\frac{a_{k}}{a^{*}}\right)+\frac{m^{2}\epsilon_{k}^{2}}{2\lambda}\|(-\Delta)^{-1/4}S\|_{L^{2}}^{2}+\epsilon_{k}\int_{\mathbb{R}^{3}}V\left(\frac{\epsilon_{k}}{\lambda}x\right)\left|S(x)\right|^{2}{\rm d}x,
	\end{align*}
	where we have used \eqref{eq:boson star inequality} and \eqref{ineq:operator} for the term in the brakets. 
	
	On the other hand, since $w_{k}$ satisfies \eqref{eq:E-L equation for w_a-1} we have
	\begin{align*}
	\epsilon_{k}\mathcal{E}_{a_{k}}(u_{k}) & =\|(-\Delta+m^{2}\epsilon_{k}^{2})^{1/4}w_{k}\|_{L^{2}}^{2}+\epsilon_{k}\int_{\mathbb{R}^{3}}V\left(\epsilon_{k}x+z\right)w_{k}(x)^{2}{\rm d}x\\
	& \quad-\frac{a_{k}}{2}\iint_{\mathbb{R}^{3}\times\mathbb{R}^{3}}\frac{\left|w_{k}(x)\right|^{2}\left|w_{k}(y)\right|^{2}}{\left|x-y\right|}{\rm d}x{\rm d}y\\
	& =\epsilon_{k}\mu_{k}\int_{\mathbb{R}^{3}}\left|w_{k}(x)\right|^{2}{\rm d}x+\frac{a_{k}}{2}\iint_{\mathbb{R}^{3}\times\mathbb{R}^{3}}\frac{\left|w_{k}(x)\right|^{2}\left|w_{k}(y)\right|^{2}}{\left|x-y\right|}{\rm d}x{\rm d}y,
	\end{align*}
	By assumption that $u_{k}$ is a minimizer of $I(a_{k})$, we have $\mathcal{E}_{a_{k}}(u_{k})\leq \mathcal{E}_{a_{k}}(s)$ and hence
	$$
	\liminf_{k\to\infty}\epsilon_{k}\mathcal{E}_{a_{k}}(u_{k})\leq \liminf_{k\to\infty}\epsilon_{k}\mathcal{E}_{a_{k}}(s).
	$$
	
	From the above estimates and by Fatou's Lemma we deduce that
	\begin{align*}
	& \frac{a^{*}}{2}\iint_{\mathbb{R}^{3}\times\mathbb{R}^{3}}\frac{\left|Q(x)\right|^{2}\left|Q(y)\right|^{2}}{\left|x-y\right|}{\rm d}x{\rm d}y=\frac{a^{*}}{2\lambda}\iint_{\mathbb{R}^{3}\times\mathbb{R}^{3}}\frac{\left|w(x)\right|^{2}\left|w(y)\right|^{2}}{\left|x-y\right|}{\rm d}x{\rm d}y\\
	& \leq\liminf_{k\to\infty}\frac{a_{k}}{2\lambda}\iint_{\mathbb{R}^{3}\times\mathbb{R}^{3}}\frac{\left|w_{k}(x)\right|^{2}\left|w_{k}(y)\right|^{2}}{\left|x-y\right|}{\rm d}x{\rm d}y\leq1.
	\end{align*}
	This inequality together with \eqref{ineq:ground state contrary 1} and \eqref{ineq:ground state contrary 2},
	imply that 
	$$
	1=\frac{a^{*}}{2}\iint_{\mathbb{R}^{3}\times\mathbb{R}^{3}}\frac{\left|Q(x)\right|^{2}\left|Q(y)\right|^{2}}{\left|x-y\right|}{\rm d}x{\rm d}y=\|(-\Delta)^{1/4}Q\|_{L^{2}}^{2}=\|Q\|_{L^{2}}^{2}.
	$$
	
	Thus we have proved that $Q\in \mathcal{G}$. 
	
	We note that $\|w\|_{L^{2}}=\|Q\|_{L^{2}}=1$. From the norm preservation, we conclude that $w_{k}\to w$ strongly
	in $L^{2}(\mathbb{R}^{3})$. In fact, $w_{k}\to w$ strongly in $L^{r}(\mathbb{R}^{3})$ for $2\leq r<3$ because of $H^{1/2}(\mathbb{R}^{3})$
	boundedness. By the Hardy-Littlewood-Sobolev inequality  we have 
	$$
	\lim_{k\to\infty}\iint_{\mathbb{R}^{3}\times\mathbb{R}^{3}}\frac{\left|w_{k}(x)\right|^{2}\left|w_{k}(y)\right|^{2}}{\left|x-y\right|}{\rm d}x{\rm d}y=\iint_{\mathbb{R}^{3}\times\mathbb{R}^{3}}\frac{\left|w(x)\right|^{2}\left|w(y)\right|^{2}}{\left|x-y\right|}{\rm d}x{\rm d}y.
	$$
	We deduce from this convergence and the inequality
	$$
	\epsilon_{k}\mathcal{E}_{a_{k}}(\epsilon_{k}^{-3/2}w_{k}(\epsilon_{k}^{-1}x-\epsilon_{k}^{-1}z_{k})) = \epsilon_{k}\mathcal{E}_{a_{k}}(u_{k}) \leq \epsilon_{k}\mathcal{E}_{a_{k}}(\epsilon_{k}^{-3/2}w(\epsilon_{k}^{-1}x))
	$$
	that
	\begin{align*}
	&  \limsup_{k\to\infty}\|(-\Delta)^{1/4}w_{k}\|_{L^{2}}^{2} \leq \limsup_{k\to\infty}\|(-\Delta+m^{2}\epsilon_{k}^{2})^{1/4}w_{k}\|_{L^{2}}^{2}\\
	& \leq \limsup_{k\to\infty}\|(-\Delta+m^{2}\epsilon_{k}^{2})^{1/4}w\|_{L^{2}}^{2} = \|(-\Delta)^{1/4}w\|_{L^{2}}^{2}.
	\end{align*}
	On the other hand, since $w_{k}\rightharpoonup w$ weakly in $H^{1/2}(\mathbb{R}^{3})$, by Fatou's Lemma we have
	\[
	\liminf_{k\to\infty}\|(-\Delta)^{1/4}w_{k}\|_{L^{2}}^{2}\geq\|(-\Delta)^{1/4}w\|_{L^{2}}^{2}.
	\]
	
	Therefore we have proved that
	$$
	\lim_{k\to\infty}\|(-\Delta)^{1/4}w_{k}\|_{L^{2}}^{2}=\|(-\Delta)^{1/4}w\|_{L^{2}}^{2},
	$$
	which implies that $w_{k}\to w$ strongly in $H^{1/2}(\mathbb{R}^{3})$.
\end{proof}

\subsection{Proof of Theorem \ref{thm:behavior-trapping}}
Let $p=\max_{1\leq i\leq n} p_i$ and we define $q=\min\{p,1\}$. 

\begin{lem}
	\label{lem:estimate of I(a_{k}) trapping} There exist positive constants
	$M_{1}<M_{2}$ independent of $a_{k}$ such that
	\begin{equation}
	M_{1}(a^{*}-a_{k})^{\frac{q}{q+1}}\leq I(a_{k})\leq M_{2}(a^{*}-a_{k})^{\frac{q}{q+1}}\label{ineq:estimate I(a_{k})}.
	\end{equation}
\end{lem}
\begin{proof}
	Since $I(a_k)$ is decreasing and uniformly bounded for $0\leq a_k\leq a^*$, it suffices to consider the case when $a_k$ is close to $a^*$. We start with proof of the upper bound in \eqref{ineq:estimate I(a_{k})}.
	We proceed similarly to the proof of Theorem \ref{thm:existence of minimizers},
	and use a trial state of the form \eqref{eq:trial state trapping-periodic}.
	Picking $x_{0}\in V^{-1}(0)$, we can choose $R$ small enough such
	that $V(x)\leq C\left|x-x_{0}\right|^{p}$ for $\left|x-x_{0}\right|\leq R$,
	in which case we have 
	$$
	\int_{\mathbb{R}^{3}}V(x)\left|u_{R,\tau}(x)\right|^{2}\leq C\frac{1}{\tau^{p}}A_{R,\tau}^{2}\int_{\mathbb{R}^{3}}\left|x\right|^{p}\left|Q(x)\right|^{2}{\rm d}x.
	$$
	We obtain that, for $\tau$ large,
	$$
	I(a_{k})\leq\frac{\tau(a^{*}-a_{k})}{a^{*}}+\frac{m^{2}}{2\tau}\|(-\Delta)^{-1/4}Q\|_{L^{2}}^{2}+C\frac{1}{\tau^{p}}A_{R,\tau}^{2}\int_{\mathbb{R}^{3}}\left|x\right|^{p}\left|Q(x)\right|^{2}{\rm d}x+o(1)_{\tau\to\infty}.
	$$
	\begin{itemize}
		\item If $0<p\leq 1$, by taking $\tau=(a^{*}-a_{k})^{-\frac{1}{p+1}}$ we arrive at the desired upper bound.

		\item If $p>1$, we choose $s=\min\{p,2\}$. It follows from decay property of $Q$ in \eqref{eq:decay} that $\int_{\mathbb{R}^{3}}\left|x\right|^{s}\left|Q(x)\right|^{2}{\rm d}x<\infty$. Since $p>s>1$, we have $V(x)\leq C\left|x-x_{0}\right|^{p}\leq C\left|x-x_{0}\right|^{s}$ for $\left|x-x_{0}\right|\leq R$ with $R$ small. Hence 
	$$
	\int_{\mathbb{R}^{3}}V(x)\left|u_{R,\tau}(x)\right|^{2}\leq C\frac{1}{\tau^{s}}A_{R,\tau}^{2}\int_{\mathbb{R}^{3}}\left|x\right|^{s}\left|Q(x)\right|^{2}{\rm d}x.
	$$
The desired upper bound is obtained by taking $\tau=(a^{*}-a_{k})^{-\frac{1}{2}}$.
		\end{itemize}
	Next we prove the lower bound in \eqref{ineq:estimate I(a_{k})}. We also consider two cases.
	
	\begin{itemize}
		\item If $0<p\leq 1$, then $q=p$. From \eqref{eq:boson star inequality} we infer that, for any $\gamma>0$,
		\begin{align*}
		I(a_{k}) & \geq\int_{\mathbb{R}^{3}}V(x)\left|u_{k}(x)\right|^{2}{\rm d}x+\left(1-\frac{a_{k}}{a^{*}}\right)\|(-\Delta)^{1/4}u_{k}\|_{L^{2}}^{2}\\
		& \geq\gamma+\int_{\mathbb{R}^{3}}\left(V(x)-\gamma\right)\left|u_{k}(x)\right|^{2}{\rm d}x+\frac{a^{*}-a_{k}}{a^{*}}S_{3}\|u_{k}\|_{L^{3}}^{2}\\
		& \geq\gamma-\int_{\mathbb{R}^{3}}\left[\gamma-V(x)\right]_{+}\left|u_{k}(x)\right|^{2}{\rm d}x+\left(\int_{\mathbb{R}^{3}}C(a^{*}-a_{k})^{3/2}\left|u_{k}(x)\right|^{3}\right)^{2/3}
		\end{align*}
		where $\left[\cdot\right]_{+}=\max\left\{ 0,\cdot\right\} $ denotes
		the positive part and $S_{3}$ is constant in Sobolev's inequality (see e.g, \cite[Theorem 8.4]{LiLo}). By H\"older's inequality we have 
		$$
		\int_{\mathbb{R}^{3}}\left[\gamma-V(x)\right]_{+}\left|u_{k}(x)\right|^{2}{\rm d}x\leq\|u_{k}\|_{L^{2}}^{2/3}\left(\int_{\mathbb{R}^{3}}\left[\gamma-V(x)\right]_{+}^{3/2}\left|u_{k}(x)\right|^{2}{\rm d}x\right)^{2/3}.
		$$
		We deduce from the above estimates, together with two inequalities 
		$$
		A^{2/3}+B^{2/3}\geq\left(A+B\right)^{2/3},\,\,A,B\geq0
		$$
		and 
		$$
		C(a^{*}-a_{k})^{3/2}\left|u_{k}(x)\right|^{3}+\frac{4\left[\gamma-V(x)\right]_{+}^{9/2}}{27C^{2}(a^{*}-a_{k})^{3}}\geq\left[\gamma-V(x)\right]_{+}^{3/2}\left|u_{k}(x)\right|^{2},
		$$
		that
		$$
		I(a_{k})\geq\gamma-\left(\int_{\mathbb{R}^{3}}\frac{4\left[\gamma-V(x)\right]_{+}^{9/2}}{27C^{2}(a^{*}-a_{k})^{3}}{\rm d}x\right)^{2/3}.
		$$
		
		For small $\gamma$, the set $\left\{ x\in\mathbb{R}^{3}:V(x)\leq\gamma\right\} $
		is contained in the disjoint union of $n$ balls of radius at most
		$k\gamma^{1/p}$ (for some $k>0$), centered at the minima $x_{i}$.
		Moreover, $V(x)\geq\left(\frac{\left|x-x_{i}\right|}{k}\right)^{p}$
		on these balls. Hence
		$$
		\int_{\mathbb{R}^{3}}\left[\gamma-V(x)\right]_{+}^{9/2}{\rm d}x\leq n\int_{\mathbb{R}^{3}}\left[\gamma-\left(\frac{\left|x\right|}{k}\right)^{p}\right]_{+}^{9/2}{\rm d}x\leq C\gamma^{3+\frac{2}{p}}.
		$$
		Thus we arrive at
		$$
		\mathcal{E}_{a_{k}}(u_{k})\geq\gamma-C\frac{\gamma^{3+\frac{2}{p}}}{(a^{*}-a_{k})^{2}}.
		$$
		
		The lower bound in \eqref{ineq:estimate I(a_{k})} therefore follows from
		the above estimate by taking $\gamma$ to be equal to $C(a^{*}-a_{k})^{\frac{p}{p+1}}$
		for a suitable constant $C>0$. 
	\end{itemize}
	
	\begin{itemize}
		\item If $p>1$, then $q=1$. From \eqref{eq:boson star inequality} and the upper bound of $I(a_{k})$ in \eqref{ineq:estimate I(a_{k})} we see that
		$$
		M_{2} (a^{*}-a_{k})^{\frac{1}{2}} \geq \left(1-\frac{a_{k}}{a^{*}}\right)\|(-\Delta+m^{2})^{1/4}u_{k}\|_{L^{2}}^{2},
		$$
		which implies that
		$$
		\|(-\Delta+m^{2})^{1/4}u_{k}\|_{L^{2}}^{2}\leq M_{2} a^{*} (a^{*}-a_{k})^{-\frac{1}{2}}.
		$$
		Again it follows from \eqref{eq:boson star inequality}, \eqref{ineq:operator} and H\"older's inequality that
		\begin{align*}
		I(a_{k}) = \mathcal{E}_{a_{k}}(u_{k}) & \geq \frac{m^{2}}{2}\|(-\Delta+m^{2})^{-1/4}u_{k}\|_{L^{2}}^{2}  \geq \frac{m^{2}\|u_{k}\|_{L^{2}}^{4}}{2\|(-\Delta+m^{2})^{1/4}u_{k}\|_{L^{2}}^{2}}\\
		& \geq \frac{m^{2}\|u_{k}\|_{L^{2}}^{4}}{2M_{2} a^{*} (a^{*}-a_{k})^{-\frac{1}{2}}} = M_{1} (a^{*}-a_{k})^{\frac{1}{2}}.
		\end{align*}
	\end{itemize}
\end{proof}

\begin{lem}
	\label{lem:estimate of direct}There exist positive constants $K_{1}<K_{2}$ independent of $a_{k}$ such that 
	\begin{equation}
	K_{1}(a^{*}-a_{k})^{-\frac{1}{q+1}}\leq\iint_{\mathbb{R}^{3}\times\mathbb{R}^{3}}\frac{\left|u_{k}(x)\right|^{2}\left|u_{k}(y)\right|^{2}}{\left|x-y\right|}{\rm d}x{\rm d}y\leq K_{2}(a^{*}-a_{k})^{-\frac{1}{q+1}}.\label{ineq:estimate of direct term}
	\end{equation}
\end{lem}
\begin{proof}
	The proof of this Lemma follows from Lemma \ref{lem:estimate of I(a_{k}) trapping} and a similar procedures as the proof of Lemma 4 in \cite{Ng-17}.
\end{proof}

Now we able to prove Theorem \ref{thm:behavior-trapping}. We consider two cases.

\subsubsection*{Case 1. If $0<p\leq 1$.}
Let $\epsilon_{k}:=(a^{*}-a_{k})^{\frac{1}{p}}>0$, we see that $\epsilon_{k}\to 0$ as $k\to\infty$. For $1\leq i\leq n$ and $x_{i}\in V^{-1}(0)$, we define $w_{k}^{i}(x):=\epsilon_{k}^{3/2}u_{k}(\epsilon_{k}x+x_{i})$
be $L^{2}$--normalized of $u_{k}$. It follows from \eqref{ineq:boundeness of I(a)} and Lemma \ref{lem:estimate of I(a_{k}) trapping} that there exists positive constant $C$ such that 
\begin{equation}
\int_{\mathbb{R}^{3}}V(x)\left|u_{k}(x)\right|^{2}{\rm d}x\leq C\epsilon_{k}^p,\quad \|(-\Delta)^{1/4}w_{k}^{i}\|_{L^{2}}^{2}\leq C.\label{ineq:estimate of minimizer-1}
\end{equation}
Thus $\left\{ w_{k}^{i}\right\} $ is bounded in $H^{1/2}(\mathbb{R}^{3})$, and hence we may assume that $w_{k}^{i}\rightharpoonup w^{i}$ weakly in $H^{1/2}(\mathbb{R}^{3})$. Extracting a subsequence if necessary, we have $w_{k}^{i}\to w^{i}$ strongly in $L_{{\rm loc}}^{r}(\mathbb{R}^3)$ for any $2\leq r<3$.

For $\epsilon_{k}$ small enough, i.e. for $a_{k}$ close to $a^{*}$,
the set $\left\{ x\in\mathbb{R}^{3}:V(x)\leq\gamma\epsilon_{k}^{p}\right\}$ is contained in $n$ disjoint balls with radius at most $C\gamma^{1/p}\epsilon_{k}$ for some $C>0$, centered at the points $x_{i}$. Thus, for any $\gamma>0$,
we have 
\begin{align*}
\frac{C}{\gamma} &\geq\frac{1}{\gamma\epsilon_{k}^{p}}\int_{\mathbb{R}^{3}}V(x)\left|u_{k}(x)\right|^{2}{\rm d}x \geq\int_{V(x)\geq\gamma\epsilon_{k}^{p}}\left|u_{k}(x)\right|^{2}{\rm d}x=1-\int_{V(x)\leq\gamma\epsilon_{k}^{p}}\left|u_{k}(x)\right|^{2}{\rm d}x\\
& \geq1-\sum_{i=1}^{n}\int_{\left|x-x_{i}\right|\leq C\gamma^{1/p}\epsilon_{k}}\left|u_{k}(x)\right|^{2}{\rm d}x=1-\sum_{i=1}^{n}\int_{\left|x\right|\leq C\gamma^{1/p}}\left|w_{k}^{i}(x)\right|^{2}{\rm d}x
\end{align*}
which implies that 
$$
1-\frac{C}{\gamma}\leq\sum_{i=1}^{n}\int_{\left|x\right|\leq C\gamma^{1/p}}\left|w_{k}^{i}(x)\right|^{2}{\rm d}x\leq1.
$$
Since $w_{k}^{i}\to w^{i}$ strongly in $L_{{\rm loc}}^{2}(\mathbb{R}^3)$, the above estimate implies that
$$
1-\frac{C}{\gamma}\leq\sum_{i=1}^{n}\int_{\left|x\right|\leq C\gamma^{1/p}}|w^{i}(x)|^{2}{\rm d}x\leq1,
$$
which holds for any $\gamma>0$. Thus we conclude that 
\begin{equation}
\sum_{i=1}^{n}\|w^{i}\|_{L^{2}}^{2}=1
\label{eq:sum of norm of w^i}
\end{equation}

We recall the Lagrange multiplier $\mu_{k}$ in \eqref{eq:Lagrange mutiplier}. It follows from Lemma \ref{lem:estimate of I(a_{k}) trapping} and \ref{lem:estimate of direct} that $\epsilon_{k}\mu_{k}$ is bounded
uniformly as $k\to\infty$, and strictly negative for $a_{k}$ close
to $a^{*}$. By passing to a subsequence if necessary, we can thus
assume that $\epsilon_{k}\mu_{k}$ converges to some number \textbf{$-\lambda<0$
}as $k\to\infty$. Thus, we deduce from Lemma \ref{lem:positive ground state solution} that $w_{k}^{i}\to w^{i}$ strongly in $H^{1/2}(\mathbb{R}^3)$ where either $w^{i}\equiv0$ or $w^{i}>0$ in $\mathbb{R}^{3}$.
In the latter case we have 
\begin{equation}
w^{i}(x)=\lambda^{3/2}Q(\lambda x)
\label{eq:norm of w^i}
\end{equation}
where $Q\in \mathcal{G}$. This implies that $\|w^i\|_{L^2}^2=1$. Because of \eqref{eq:sum of norm of w^i}, we see that there is exactly one $w^{i}$ which is of the form \eqref{eq:norm of w^i}, while all the others are zero. Let $1\leq j\leq n$ be such that $w^{j}(x)=\lambda^{3/2}Q(\lambda x)$. Then we have  $p_j=p=\max_{1\leq i \leq n}p_i$. Indeed, if $p_j <p$, by Fatou's Lemma we have, for any large constant $M>0$
\begin{align}
&\lim_{k\to\infty}\frac{1}{\epsilon_{k}^{p}}\int_{\mathbb{R}^{3}}V(\epsilon_{k}x+x_{j})|w_{k}^{j}(x)|^{2}{\rm d}x \geq \int_{\mathbb{R}^{3}}\lim_{k\to\infty}\epsilon_{k}^{p_j-p}\kappa_{j}\left|x\right|^{p_j}\left|w^{j}(x)\right|^{2}{\rm d}x\nonumber\\
& =\lim_{k\to\infty}\epsilon_{k}^{p_j-p}\kappa_{j}\frac{1}{\lambda^{p_j}}\int_{\mathbb{R}^{3}}\left|x\right|^{p_j}\left|Q(x)\right|^{2}{\rm d}x\geq M,
\end{align}
where 
$$
\kappa_{j}:=\lim_{x\to x_{j}}\frac{V(x)}{\left|x-x_{j}\right|^{p_j}}>0.
$$
From \eqref{eq:boson star inequality}  and the fact that $u_{k}=\epsilon_{k}^{-3/2}w_{k}^{j}(\epsilon_{k}^{-1}x-\epsilon_{k}^{-1}x_{j})$ is a minimizer of $I(a_{k})$, we have
\begin{align}
I(a_{k}) & =\frac{1}{\epsilon_{k}}\left(\|(-\Delta+m^{2}\epsilon_{k}^{2})^{1/4}w_{k}^{j}\|_{L^{2}}^{2}-\frac{a^{*}}{2}\iint_{\mathbb{R}^{3}\times\mathbb{R}^{3}}\frac{|w_{k}^{j}(x)|^{2}|w_{k}^{j}(y)|^{2}}{\left|x-y\right|}{\rm d}x{\rm d}y\right)\nonumber\\
& \quad+\frac{a^{*}-a_{k}}{2\epsilon_{k}}\iint_{\mathbb{R}^{3}\times\mathbb{R}^{3}}\frac{|w_{k}^{j}(x)|^{2}|w_{k}^{j}(y)|^{2}}{\left|x-y\right|}{\rm d}x{\rm d}y+\int_{\mathbb{R}^{3}}V(\epsilon_{k}x+x_{j})|w_{k}^{j}(x)|^{2}{\rm d}x\nonumber\\
& \geq C_1 \frac{a^{*}-a_{k}}{\epsilon_{k}}+M\epsilon_{k}^{p}\geq C_{2}M^{\frac{1}{p+1}}\left(a^{*}-a_{k}\right)^{\frac{p}{p+1}}.
\end{align}
The above lower bound of $I(a_k)$ holds true for any $0 < p \leq 1$ and for an arbitrary large $M > 0$. This contradicts the upper bound of $I(a_k)$ in \eqref{ineq:estimate I(a_{k})}. Hence $p_j=p=\max_{1\leq i \leq n}p_i$.

Now we compute the exact value of $\lambda$. Since $u_{k}=\epsilon_{k}^{-3/2}w_{k}^{j}(\epsilon_{k}^{-1}x-\epsilon_{k}^{-1}x_{j})$ is a minimizer of $I(a_{k})$ we have
\begin{align}
I(a_{k}) & =\frac{1}{\epsilon_{k}}\left(\|(-\Delta+m^{2}\epsilon_{k}^{2})^{1/4}w_{k}^{j}\|_{L^{2}}^{2}-\frac{a^{*}}{2}\iint_{\mathbb{R}^{3}\times\mathbb{R}^{3}}\frac{|w_{k}^{j}(x)|^{2}|w_{k}^{j}(y)|^{2}}{\left|x-y\right|}{\rm d}x{\rm d}y\right)\nonumber\\
& \quad+\frac{a^{*}-a_{k}}{2\epsilon_{k}}\iint_{\mathbb{R}^{3}\times\mathbb{R}^{3}}\frac{|w_{k}^{j}(x)|^{2}|w_{k}^{j}(y)|^{2}}{\left|x-y\right|}{\rm d}x{\rm d}y+\int_{\mathbb{R}^{3}}V(\epsilon_{k}x+x_{j})|w_{k}^{j}(x)|^{2}{\rm d}x\nonumber\\
& \geq m^{2}\epsilon_{k}\|((-\Delta+m^{2}\epsilon_{k}^{2})^{1/2}+(-\Delta)^{1/2})^{-1/2}w_{k}^{j}\|_{L^{2}}^{2} \label{lim:I(a_{k})}\\
& \quad+\frac{\epsilon_{k}^{q}}{2}\iint_{\mathbb{R}^{3}\times\mathbb{R}^{3}}\frac{|w_{k}^{j}(x)|^{2}|w_{k}^{j}(y)|^{2}}{\left|x-y\right|}{\rm d}x{\rm d}y+\int_{\mathbb{R}^{3}}V(\epsilon_{k}x+x_{j})|w_{k}^{j}(x)|^{2}{\rm d}x\nonumber,
\end{align}
where we have used
\eqref{eq:boson star inequality} and \eqref{ineq:operator} for the term in the brackets.

By Fatou's Lemma, we have 
\begin{align}
& \liminf_{k\to\infty}\|((-\Delta+m^{2}\epsilon_{k}^{2})^{1/2}+(-\Delta)^{1/2})^{-1/2}w_{k}^{j}\|_{L^{2}}^{2}\nonumber\\
& \geq\frac{1}{2}\|(-\Delta)^{-1/4}w^{j}\|_{L^{2}}^{2}=\frac{1}{2\lambda}\|(-\Delta)^{-1/4}Q\|_{L^{2}}^{2}.\label{lim:kinetic}
\end{align}
We note that, for any $Q\in \mathcal{G}$, by H\"oder's inequality we have
$$
\|(-\Delta)^{-1/4}Q\|_{L^{2}}\geq \frac{\|Q\|_{L^{2}}^{2}}{\|(-\Delta)^{1/4}Q\|_{L^{2}}}=1,
$$
which implies that $\inf_{Q\in \mathcal{G}}\|(-\Delta)^{-1/4}Q\|_{L^{2}}$ is well defined.

By the Hardy-Littlewood-Sobolev inequality, we have 
\begin{align}
& \lim_{k\to\infty}\iint_{\mathbb{R}^{3}\times\mathbb{R}^{3}}\frac{|w_{k}^{j}(x)|^{2}|w_{k}^{j}(y)|^{2}}{\left|x-y\right|}{\rm d}x{\rm d}y=\iint_{\mathbb{R}^{3}\times\mathbb{R}^{3}}\frac{\left|w^{j}(x)\right|^{2}\left|w^{j}(y)\right|^{2}}{\left|x-y\right|}{\rm d}x{\rm d}y\nonumber\\
& =\lambda\iint_{\mathbb{R}^{3}\times\mathbb{R}^{3}}\frac{\left|Q(x)\right|^{2}\left|Q(y)\right|^{2}}{\left|x-y\right|}{\rm d}x{\rm d}y=\frac{2\lambda}{a^{*}}.\label{lim:direct}
\end{align}

By Fatou's Lemma and keep in mind that $0<p_j=p\leq 1$, we have 
\begin{align} \liminf_{k\to\infty}\frac{1}{\epsilon_{k}^{p}}\int_{\mathbb{R}^{3}}V(\epsilon_{k}x+x_{j})|w_{k}^{j}(x)|^{2}{\rm d}x \geq \kappa_{j}\frac{1}{\lambda^{p}}\int_{\mathbb{R}^{3}}\left|x\right|^{p}\left|Q(x)\right|^{2}{\rm d}x.\label{lim:V}
\end{align}

\begin{itemize}
	\item If $p<1$, then it follows from \eqref{lim:I(a_{k})}, \eqref{lim:direct} and \eqref{lim:V} that
	$$
	\liminf_{k\to\infty}\frac{I(a_{k})}{\epsilon_{k}^{p}}\geq\frac{1}{a^{*}}\left(\lambda+\frac{\lambda_{\kappa}^{p+1}}{p\lambda^{p}}\right)
	$$
	where  
	$$
	\lambda_{\kappa}=\left(pa^{*}\kappa\int_{\mathbb{R}^{3}}\left|x\right|^{p}\left|Q(x)\right|{\rm d}x\right)^{\frac{1}{p+1}} \text{ and }\kappa=\min\{\kappa_i:1\leq i\leq n\}.
	$$
	Thus
	\begin{equation}\label{liminf: I(a_{k})/epsilon-1}
	\liminf_{k\to\infty}\frac{I(a_{k})}{\epsilon_{k}^{p}}\geq\inf_{\tilde{\lambda}>0}\frac{1}{a^{*}}\left(\tilde{\lambda}+\frac{\lambda_{\kappa}^{p+1}}{p\tilde{\lambda}^{p}}\right) =\frac{p+1}{p}\cdot\frac{\lambda_{\kappa}}{a^{*}}.
	\end{equation}

	\item If $p=1$, then it follows from \eqref{lim:I(a_{k})}, \eqref{lim:kinetic}, \eqref{lim:direct} and \eqref{lim:V} that
	$$
	\liminf_{k\to\infty}\frac{I(a_{k})}{\epsilon_{k}}\geq\frac{1}{a^{*}}\left(\lambda+\frac{\lambda_{\kappa}^{2}}{\lambda}\right)
	$$
	where
	$$
	\lambda_{\kappa}=\left(\frac{m^{2}a^{*}}{2}\|(-\Delta)^{-1/4}Q\|_{L^{2}}^{2}+a^{*}\kappa\int_{\mathbb{R}^{3}}\left|x\right|\left|Q(x)\right|^{2}{\rm d}x\right)^{\frac{1}{2}}
	$$
	and $\kappa=\min\{\kappa_i:1\leq i\leq n\}$.
	
	Thus
	\begin{equation}\label{liminf: I(a_{k})/epsilon-2}
	\liminf_{k\to\infty}\frac{I(a_{k})}{\epsilon_{k}}\geq\inf_{\tilde{\lambda}>0}\frac{1}{a^{*}}\left(\tilde{\lambda}+\frac{\lambda_{\kappa}^{2}}{\tilde{\lambda}}\right)= 2\frac{\lambda_{\kappa}}{a^{*}}.
	\end{equation}
\end{itemize}

Now we prove the matching upper bound in \eqref{liminf: I(a_{k})/epsilon-1}--\eqref{liminf: I(a_{k})/epsilon-2}. We take 
\begin{equation}\label{trial-state}
u_{k}(x)=\left(\frac{\tilde{\lambda}}{\epsilon_{k}}\right)^{3/2}W\left(\frac{\tilde{\lambda}}{\epsilon_{k}}\left(x-x_{i}\right)\right)
\end{equation}
as trial state for $\mathcal{E}_{a_{k}}$, where $x_{i}\in\mathcal{Z}$, i.e $p_i=p$, $\kappa_i=\kappa$, $\tilde{\lambda}>0$ and $W\in \mathcal{G}$. We use \eqref{eq:boson star inequality} and \eqref{ineq:operator} to obtain
$$
I(a_{k})  \leq \frac{m^{2}\epsilon_{k}}{2\tilde{\lambda}}\|(-\Delta)^{-1/4}W\|_{L^{2}}^{2}+\frac{\epsilon_{k}^{p}\tilde{\lambda}}{a^{*}}+\int_{\mathbb{R}^{3}}V\left(\frac{\epsilon_{k}}{\tilde{\lambda}}x+x_{i}\right)\left|W(x)\right|^{2}{\rm d}x.
$$

\begin{itemize}
	\item If $p<1$, then we have
	$$
	\limsup_{k\to\infty}\frac{I(a_{k})}{\epsilon_{k}^{p}}\leq\frac{1}{a^{*}}\left(\tilde{\lambda}+\frac{\tilde{\lambda}_\kappa^{p+1}}{p\tilde{\lambda}^{p}}\right),
	$$
	where 
	$$
	\tilde{\lambda}_\kappa=\left(pa^{*}\kappa\int_{\mathbb{R}^{3}}\left|x\right|^{p}\left|W(x)\right|^2{\rm d}x\right)^{\frac{1}{p+1}}.
	$$
	Thus, taking the infimum over $W\in \mathcal{G}$ and $\tilde{\lambda}>0$ we see that 
	\begin{equation}\label{limsup: I(a_{k})/epsilon-1}
	\limsup_{k\to\infty}\frac{I(a_{k})}{\epsilon_{k}^{p}}\leq\frac{p+1}{p}\cdot\frac{\inf_{W\in \mathcal{G}}\tilde{\lambda}_\kappa}{a^{*}}\leq \frac{p+1}{p}\cdot\frac{\lambda_\kappa}{a^{*}}.
	\end{equation}

	\item If $p=1$, then we have
	$$
	\limsup_{k\to\infty}\frac{I(a_{k})}{\epsilon_{k}}\leq\frac{1}{a^{*}}\left(\tilde{\lambda}+\frac{\tilde{\lambda}_\kappa^{2}}{\tilde{\lambda}}\right)
	$$
	where 
	$$
	\tilde{\lambda}_\kappa=\left(\frac{m^{2}a^{*}}{2}\|(-\Delta)^{-1/4}W\|_{L^{2}}^{2}+a^{*}\kappa\int_{\mathbb{R}^{3}}\left|x\right|\left|W(x)\right|^{2}{\rm d}x\right)^{\frac{1}{2}}.
	$$
	Thus, taking the infimum over $W\in \mathcal{G}$ and $\tilde{\lambda}>0$ we see that 
	\begin{equation}\label{limsup: I(a_{k})/epsilon-2}
	\limsup_{k\to\infty}\frac{I(a_{k})}{\epsilon_{k}}\leq2\frac{\inf_{W\in \mathcal{G}}\tilde{\lambda}_\kappa}{a^{*}}\leq 2\frac{\lambda_\kappa}{a^{*}}.
	\end{equation}
\end{itemize}

From \eqref{liminf: I(a_{k})/epsilon-1}-\eqref{liminf: I(a_{k})/epsilon-2} and \eqref{limsup: I(a_{k})/epsilon-1}-\eqref{limsup: I(a_{k})/epsilon-2} we conclude that, for $0<p\leq 1$,
$$
\lim_{k\to\infty}\frac{I(a_{k})}{\epsilon_{k}^{p}} = \frac{p+1}{p}\cdot\frac{\lambda}{a^{*}},
$$
and $$\lambda=\inf_{W\in \mathcal{G}}\tilde{\lambda}_\kappa=\lambda_\kappa,\quad 
\kappa_{j}=\min_{i:p_i=p}\kappa_{i}=\kappa
$$
which, together with $p_j=\min\{p_i:1\leq i\leq n\}=p$, implies that $x_j\in\mathcal{Z}$.
\subsubsection*{Case 2. If $p>1$.}
Let $\epsilon_{k}:=(a^{*}-a_{k})^{\frac{1}{2}}>0$, we see that $\epsilon_{k}\to 0$ as $k\to\infty$. We define $\tilde{w}_{k}(x):=\epsilon_{k}^{3/2}u_{k}(\epsilon_{k}x)$
be $L^{2}$--normalized of $u_{k}$. It follows from Lemma \ref{lem:estimate of direct} that
\begin{equation}
0<K_{1}\leq \iint_{\mathbb{R}^{3}\times\mathbb{R}^{3}}\frac{\left|\tilde{w}_{k}(x)\right|^{2}\left|\tilde{w}_{k}(y)\right|^{2}}{\left|x-y\right|}{\rm d}x{\rm d}y\leq K_{2}.\label{ineq:vanishing}
\end{equation}

We claim that there exist a sequence $\{y_{k}\}\subset \mathbb{R}^3$ and positive constant $R_{0}$ such that 
\begin{equation}
\liminf_{k\to\infty} \int_{B(y_{k},R_{0})}\left|\tilde{w}_{k}(x)\right|^{2}{\rm d}x>0.\label{eq:vanishing2}
\end{equation}
On the contrary, we assume that for any $R$ there exists a subsequence of $\left\{ a_{k}\right\} $,
still denoted by $\left\{ a_{k}\right\} $, such that 
\begin{equation} \label{eq:contr}
\lim_{k\to\infty}\sup_{y\in\mathbb{R}^3}\int_{B(y,R)}\left|\tilde{w}_{k}(x)\right|^{2}{\rm d}x=0.
\end{equation}
It follows from \cite[Lemma 9]{LeLe-11}
that $\lim_{k\to\infty}\|\tilde{w}_{k}\|_{L^{r}}=0$ for any $2\leq r<3$.
By the Hardy-Littlewood-Sobolev inequality we have 
\[
\iint_{\mathbb{R}^{3}\times\mathbb{R}^{3}}\frac{\left|\tilde{w}_{k}(x)\right|^{2}\left|\tilde{w}_{k}(y)\right|^{2}}{\left|x-y\right|}{\rm d}x{\rm d}y\to0
\]
as $k\to\infty$, which contradicts to  \eqref{ineq:vanishing}. Thus \eqref{eq:contr} does not occurs, and hence \eqref{eq:vanishing2} holds true. 

Let $w_{k}$ be non-negative $L^{2}$-normalize of $u_{k}$, defined by $$w_{k}(x)=\tilde{w}_{k}(x+y_{k})=\epsilon_{k}^{3/2}u_{k}(\epsilon_{k}x+\epsilon_{k}y_{k}).$$
It follows from \eqref{ineq:boundeness of I(a)} and Lemma \ref{lem:estimate of I(a_{k}) trapping} that there exists positive constant $C$ such that 
$$
\|(-\Delta)^{1/4}w_{k}\|_{L^{2}}^{2}\leq C.
$$
Thus ${w_{k}}$ is bounded in $H^{1/2}(\mathbb{R}^{3})$, and hence we may assume that $w_{k}\rightharpoonup w$ weakly in $H^{1/2}(\mathbb{R}^{3})$ and $w_{k}\to w$ pointwise almost everywhere. Extracting a subsequence if necessary, we have $w_{k}\to w$ strongly in $L^{r}_{{\rm loc}}(\mathbb{R}^{3})$ for $2\leq r<3$.

We recall the Lagrange multiplier $\mu_{k}$ in \eqref{eq:Lagrange mutiplier}. It follows from Lemma \ref{lem:estimate of I(a_{k}) trapping} and \ref{lem:estimate of direct} that $\epsilon_{k}\mu_{k}$ is bounded
uniformly as $k\to\infty$, and strictly negative for $a_{k}$ close
to $a^{*}$. By passing to a subsequence if necessary, we can thus
assume that $\epsilon_{k}\mu_{k}$ converges to some number \textbf{$-\lambda<0$
}as $k\to\infty$. Thus, we deduce from Lemma \ref{lem:positive ground state solution} that $w_{k}\to w$ strongly in $H^{1/2}(\mathbb{R}^3)$ where either $w\equiv0$ or $w>0$ in $\mathbb{R}^{3}$.
In the latter case we have 
$$
w(x)=\lambda^{3/2}Q(\lambda x)
\label{eq:norm of w^i}
$$
where $Q\in \mathcal{G}$. We infer from \eqref{eq:vanishing2} and the convergence of $w_{k}$ in $L^{2}_{{\rm loc}}(\mathbb{R}^{3})$ that
$$
\int_{B(0,R_{0})}\left|w(x)\right|^{2}{\rm d}x=\lim_{k\to\infty} \int_{B(0,R_{0})}\left|w_{k}(x)\right|^{2}{\rm d}x=\lim_{k\to\infty} \int_{B(y_{k},R_{0})}\left|\tilde{w}_{k}(x)\right|^{2}{\rm d}x>0.
$$
This implies that $w\not \equiv 0$, and hence $Q\not \equiv 0$.

Now we compute the exact value of $\lambda$. Since $u_{k}=\epsilon_{k}^{-3/2}w_{k}(\epsilon_{k}^{-1}x-y_{k})$ is a minimizer of $I(a_{k})$ we have
\begin{align}
I(a_{k}) &\geq m^{2}\epsilon_{k}\|((-\Delta+m^{2}\epsilon_{k}^{2})^{1/2}+(-\Delta)^{1/2})^{-1/2}w_{k}\|_{L^{2}}^{2} \nonumber\\
& \quad+\frac{\epsilon_{k}}{2}\iint_{\mathbb{R}^{3}\times\mathbb{R}^{3}}\frac{|w_{k}(x)|^{2}|w_{k}(y)|^{2}}{\left|x-y\right|}{\rm d}x{\rm d}y.\label{lim:I(a_{k})-2}
\end{align}

By Fatou's Lemma, we have 
\begin{align}
& \liminf_{k\to\infty}\|((-\Delta+m^{2}\epsilon_{k}^{2})^{1/2}+(-\Delta)^{1/2})^{-1/2}w_{k}\|_{L^{2}}^{2}\nonumber\\
& \geq\frac{1}{2}\|(-\Delta)^{-1/4}w_{k}\|_{L^{2}}^{2}=\frac{1}{2\lambda}\|(-\Delta)^{-1/4}Q\|_{L^{2}}^{2}.\label{lim:kinetic-2}
\end{align}

By the Hardy-Littlewood-Sobolev inequality, we have 
\begin{align}
& \lim_{k\to\infty}\iint_{\mathbb{R}^{3}\times\mathbb{R}^{3}}\frac{|w_{k}(x)|^{2}|w_{k}(y)|^{2}}{\left|x-y\right|}{\rm d}x{\rm d}y=\iint_{\mathbb{R}^{3}\times\mathbb{R}^{3}}\frac{\left|w_{k}(x)\right|^{2}\left|w_{k}(y)\right|^{2}}{\left|x-y\right|}{\rm d}x{\rm d}y\nonumber\\
& =\lambda\iint_{\mathbb{R}^{3}\times\mathbb{R}^{3}}\frac{\left|Q(x)\right|^{2}\left|Q(y)\right|^{2}}{\left|x-y\right|}{\rm d}x{\rm d}y=\frac{2\lambda}{a^{*}}.\label{lim:direct-2}
\end{align}

It follows from \eqref{lim:I(a_{k})}, \eqref{lim:kinetic} and \eqref{lim:direct} that
	$$
	\liminf_{k\to\infty}\frac{I(a_{k})}{\epsilon_{k}}\geq\frac{\lambda}{a^{*}}+\frac{m^{2}}{2\lambda}\|(-\Delta)^{-1/4}Q\|_{L^{2}}^{2}.
	$$
	Thus
	\begin{equation}\label{liminf: I(a_{k})/epsilon-3}
	\liminf_{k\to\infty}\frac{I(a_{k})}{\epsilon_{k}}\geq\inf_{\tilde{\lambda}>0}\left(\frac{\tilde{\lambda}}{a^{*}}+\frac{m^{2}}{2\tilde{\lambda}}\|(-\Delta)^{-1/4}Q\|_{L^{2}}^{2}\right)=m\sqrt{\frac{2}{a^{*}}}\|(-\Delta)^{-1/4}Q\|_{L^{2}}.
	\end{equation}

Now we prove the matching upper bound in \eqref{liminf: I(a_{k})/epsilon-3}, we take a trial state of the form \eqref{trial-state} for $\mathcal{E}_{a_{k}}$, and minimizes over $1 \leq i \leq n$ such that $p_i=p$, $\kappa_i=\kappa$. Since $p>1$, we obtain that
	$$
	\limsup_{k\to\infty}\frac{I(a_{k})}{\epsilon_{k}}\leq\frac{\tilde{\lambda}}{a^{*}}+\frac{m^{2}}{2\tilde{\lambda}}\|(-\Delta)^{-1/4}W\|_{L^{2}}^{2}.
	$$
	Thus, taking the infimum over $W\in \mathcal{G}$ and $\tilde{\lambda}>0$ we see that 
	\begin{equation}\label{limsup: I(a_{k})/epsilon-3}
	\limsup_{k\to\infty}\frac{I(a_{k})}{\epsilon_{k}}\leq m\sqrt{\frac{2}{a^{*}}}\inf_{W\in \mathcal{G}}\|(-\Delta)^{-1/4}W\|_{L^{2}}=m\sqrt{\frac{2}{a^{*}}}\|(-\Delta)^{-1/4}Q\|_{L^{2}}.
	\end{equation}

From \eqref{liminf: I(a_{k})/epsilon-3}-\eqref{limsup: I(a_{k})/epsilon-3} we conclude that
$$
\lim_{k\to\infty}\frac{I(a_{k})}{\epsilon_{k}} = 2\frac{\lambda}{a^{*}},
$$
and 
$$\lambda= m\sqrt{\frac{a^{*}}{2}}\|(-\Delta)^{-1/4}Q\|_{L^{2}} = m\sqrt{\frac{a^{*}}{2}}\inf_{W\in \mathcal{G}}\|(-\Delta)^{-1/4}W\|_{L^{2}}
$$

We note that, in any case, the limit of $I(a_{k})/\epsilon_{k}^{q}$, where $q=\min\{p,1\}$, is independent of the subsequence $\{a_{k}\}$. Therefore, we have the convergence of the whole family in \eqref{lim: I(a_k)/epsilon_k}. 

\subsection{Proof of Theorem \ref{thm:behavior-periodic}}

Let $p>0$, and we define $q=\min\{p,1\}$.

First, we claim that there exist postitive constants $M_{1}<M_{2}$ independent of $\{a_{k}\}$ such that
\begin{equation}\label{ineq: upper bound I(a_{k}) periodic bad}
M_{1}(a^{*}-a_{k})\leq I(a_{k})\leq M_{2}(a^{*}-a_{k})^{\frac{q}{q+1}}.
\end{equation}
Indeed, by using the trial state of the form \eqref{eq:trial state trapping-periodic}, the proof of upper bound in \eqref{ineq: upper bound I(a_{k}) periodic bad} follows from the proof of upper bound in \eqref{ineq:estimate I(a_{k})}. The lower bound in \eqref{ineq: upper bound I(a_{k}) periodic bad} follows from \eqref{eq:boson star inequality} and the fact that $\|u_{k}\|_{L^2}^{2}=1$,
$$
I(a_{k})=\mathcal{E}_{a_{k}}(u_{k})\geq\left(1-\frac{a_{k}}{a^{*}}\right)\|(-\Delta+m^{2})^{1/4}u_{k}\|_{L^{2}}^{2}\geq m\frac{a^{*}-a_{k}}{a^{*}}.
$$

Now we use the estimates of $I(a_{k})$ in \eqref{ineq: upper bound I(a_{k}) periodic bad} to prove that there exists positive constant $K$ independent of $a_{k}$ such that, as $k\to\infty$, we have
\begin{equation}\label{ineq:estimate of direct term-1}
\iint_{\mathbb{R}^{3}\times\mathbb{R}^{3}}\frac{\left|u_{k}(x)\right|^{2}\left|u_{k}(y)\right|^{2}}{\left|x-y\right|}{\rm d}x{\rm d}y\geq K.
\end{equation}
Indeed, for any $b$ such that $0<b<a_{k}$ with $b=a_{k}-\left(a^{*}-a_{k}\right)^{\epsilon}$
with $\epsilon<\frac{q}{q+1}$, we have 
$$
I(b)\leq\mathcal{E}_{b}(u_{k})=I(a_{k})+\frac{a_{k}-b}{2}\iint_{\mathbb{R}^{3}\times\mathbb{R}^{3}}\frac{\left|u_{k}(x)\right|^{2}\left|u_{k}(y)\right|^{2}}{\left|x-y\right|}{\rm d}x{\rm d}y.
$$
We deduce from the above inequality and the estimates of $I(a_{k})$ in \eqref{ineq: upper bound I(a_{k}) periodic bad} that there exist two positive
constants $M_{1}<M_{2}$ such that for any $0<b<a_{k}<a^{*}$, 
\begin{align*}
\frac{1}{2}\iint_{\mathbb{R}^{3}\times\mathbb{R}^{3}}\frac{\left|u_{k}(x)\right|^{2}\left|u_{k}(y)\right|^{2}}{\left|x-y\right|}{\rm d}x{\rm d}y
&
\geq\frac{I(b)-I(a_{k})}{a_{k}-b} \geq\frac{M_{1}\left(a^{*}-b\right)-M_{2}\left(a^{*}-a_{k}\right)^{\frac{q}{q+1}}}{a_{k}-b}\\
&
=M_{1}+M_{1}(a^{*}-a_{k})^{1-\epsilon}-M_{2}(a^{*}-a_{k})^{\frac{q}{q+1}-\epsilon}.
\end{align*}
Since $\epsilon<\frac{q}{q+1}$ and $M_{2}>M_{1}>0$, we have $M_{1}(a^{*}-a_{k})^{1-\epsilon}-M_{2}(a^{*}-a_{k})^{\frac{q}{q+1}-\epsilon}\to0$
as $k\to\infty$. Hence \eqref{ineq:estimate of direct term-1} holds true.

\begin{lem}\label{lem: important periodic} We define $\epsilon_{k}^{-1}:=\|(-\Delta)^{1/4}u_{k}\|_{L^{2}}^{2}$. Then the following statements hold true
	\begin{itemize}
		\item[(i)] $\epsilon_{k}\to0$ as $k\to\infty$.
		
		\item[(ii)] There exist a sequence $\left\{ y_{k}\right\}\subset\mathbb{R}^{3}$ and a postitive constant $R_{1}$ such that the sequence $w_{k}(x)=\epsilon_{k}^{3/2}u_{k}\left(\epsilon_{k}x+\epsilon_{k}y_{k}\right)$, satisfies 
		\begin{equation}\label{eq:vanishing6}
		\liminf_{k\to\infty}\int_{B(0,R_{1})}\left|w_{k}(x)\right|^{2}{\rm d}x>0.
		\end{equation}
		
		\item[(iii)] We write $\epsilon_{k}y_{k}=z_{k}+x_{k}$ for $z_{k}\in \mathbb{Z}^3$ and $x_{k}\in [0,1]^3$. Then there exists a subsequence of $\{a_{k}\}$, still denoted by $\{a_{k}\}$, such that $x_{k}\to x_{0}$ where $x_{0}\in\mathbb{R}^{3}$ being a minimum point of $V$, i.e. $V(x_{0})=0$.
	\end{itemize}
	
	Futhermore, we have $w_{k}(x)\to Q(x)$ strongly in $H^{1/2}(\mathbb{R}^{3})$ as $k\to\infty$, where $Q\in \mathcal{G}$. 
\end{lem}
\begin{proof}
	(i) On the contrary we assume that there exists a subsequence $\left\{ a_{k}\right\}$, still denoted by $\left\{ a_{k}\right\} $, such that $\left\{ u_{k}\right\} $ is uniformly bounded in $H^{1/2}(\mathbb{R}^{3})$. Thus, there exists $u\in H^{1/2}(\mathbb{R}^{3})$ such that $u_{k}\rightharpoonup u$ weakly in $H^{1/2}(\mathbb{R}^{3})$ and $u_{k}\to u$ a.e. in $\mathbb{R}^3$. Extracting a subsequence if necessary, we have $u_{k}\to u$ strong in $L_{{\rm loc}}^{r}(\mathbb{R}^3)$ for $2\leq r<3$. We claim that there exists sequence $\left\{ y_{k}\right\}\subset \mathbb{R}^3$ and positive constant $R_{0}$ such that 
	\begin{equation}
	\liminf_{k\to\infty}\int_{B(y_{k},R_{0})}\left|u_{k}(x)\right|^{2}{\rm d}x>0.\label{eq:vanishing4}
	\end{equation}
	On the contrary, we assume that for any $R>0$ there exists a sequence of $\left\{ a_{k}\right\}$, still denoted by $\left\{ a_{k}\right\}$, such that 
	$$
	\lim_{k\to\infty}\sup_{y\in\mathbb{R}^3}\int_{B(y,R)}\left|u_{k}(x)\right|^{2}{\rm d}x=0.
	$$
	It follows from \cite[Lemma 9]{LeLe-11} that $\lim_{k\to\infty}\|u_{k}\|_{L^{q}}\to0$ for $2\leq q<3$. By the Hardy-Littlewood-Sobolev inequality we have 
	$$
	\lim_{k\to\infty}\iint_{\mathbb{R}^{3}\times\mathbb{R}^{3}}\frac{\left|u_{k}(x)\right|^{2}\left|u_{k}(y)\right|^{2}}{\left|x-y\right|}{\rm d}x{\rm d}y=0.
	$$
	This contradicts \eqref{ineq:estimate of direct term-1}. Hence \eqref{eq:vanishing4} holds true. 
	
	From \eqref{eq:vanishing4} and using the same arguments of the proof of Theorem \ref{thm:existence of minimizers} (iii) in Section \ref{subsec:Periodic-potential} (with $a$ replaced by $a_{k}$ and using $\lim_{k\to\infty}I(a_{k})=I(a^*)$), we also arrive at a contradiction to Theorem \ref{thm:existence of minimizers} (ii). Hence, we conclude that $\epsilon_{k}\to0$ as $k\to\infty$. 
	
	(ii) We define $\tilde{w}_{k}(x):=\epsilon_{k}^{3/2}u_{k}\left(\epsilon_{k}x\right)$ be $L^{2}$--normalized of $u_{k}$. From \eqref{eq:boson star inequality} we have 
	$$
	0\leq \|(-\Delta)^{1/4}u_{k}\|_{L^{2}}^{2}-\frac{a_{k}}{2}\iint_{\mathbb{R}^{3}\times\mathbb{R}^{3}}\frac{\left|u_{k}(x)\right|^{2}\left|u_{k}(y)\right|^{2}}{\left|x-y\right|}{\rm d}x{\rm d}y\leq \mathcal{E}_{a_{k}}(u_{k}) = I(a_{k}).
	$$
	Thus, we deduce from the upper bound of $I(a_{k})$ in
	\eqref{ineq: upper bound I(a_{k}) periodic bad} that 
	\begin{equation}\label{eq:convergence of direct term and epsilon 1}
	\iint_{\mathbb{R}^{3}\times\mathbb{R}^{3}}\frac{\left|\tilde{w}_{k}(x)\right|^{2}\left|\tilde{w}_{k}(y)\right|^{2}}{\left|x-y\right|}{\rm d}x{\rm d}y = \epsilon_{k}\iint_{\mathbb{R}^{3}\times\mathbb{R}^{3}}\frac{\left|u_{k}(x)\right|^{2}\left|u_{k}(y)\right|^{2}}{\left|x-y\right|}{\rm d}x{\rm d}y\to\frac{2}{a^{*}}
	\end{equation}
	as $k\to \infty$.
	
	By the same arguments of proof of \eqref{eq:vanishing4} (in replacing contradiction \eqref{ineq:estimate of direct term-1} by \eqref{eq:convergence of direct term and epsilon 1}), we can prove that there exists sequence $\left\{ y_{k}\right\}\subset \mathbb{R}^3$
	and positive constant $R_{1}$ such that 
	\begin{equation}
	\liminf_{k\to\infty}\int_{B(y_{k},R_{1})}\left|\tilde{w}_{k}(x)\right|^{2}{\rm d}x>0.\label{eq:vanishing5}
	\end{equation}
	Therefore, \eqref{eq:vanishing6} follows from \eqref{eq:vanishing5} and definition of $w_{k}$ and $\tilde{w}_{k}$.
	
	(iii) 	On the contrary, we assume that $V(x_{0})>0$. We know $V$ is continuous periodic function. Thus, we deduce from \eqref{eq:vanishing6} and Fatou's lemma that 
	$$
	\lim_{k\to\infty}\int_{\mathbb{R}^{3}}V(\epsilon_{k}x+\epsilon_{k}y_{k})\left|w_{k}(x)\right|^{2}{\rm d}x\geq V(x_{0})\lim_{k\to\infty}\int_{\mathbb{R}^{3}}\left|w_{k}(x)\right|^{2}{\rm d}x>0
	$$
	which is impossible since 
	$$
	0\leq\int_{\mathbb{R}^{3}}V\left(\epsilon_{k}x+\epsilon_{k}y_{k}\right)\left|w_{k}(x)\right|^{2}{\rm d}x=\int_{\mathbb{R}^{3}}V(x)\left|u_{k}(x)\right|^{2}{\rm d}x\leq I(a_{k}),
	$$
	and $I(a_{k})\to0$ as $k\to \infty$, thanks to the upper bound of $I(a_{k})$ in \eqref{ineq: upper bound I(a_{k}) periodic bad}.
	
	Thus $x_{k}\to x_{0}$ as $k\to\infty$ where $x_{0}\in\mathbb{R}^{3}$ such that $V(x_{0})=0$.
	
	Finally, we recall the Lagrange multiplier $\mu_{k}$ in \eqref{eq:Lagrange mutiplier}. It follows from \eqref{eq:convergence of direct term and epsilon 1} and the upper bound of $I(a_{k})$ in \eqref{ineq: upper bound I(a_{k}) periodic bad} that $\epsilon_{k}\mu_{k}\to-1$ as $k\to\infty$. Note that, by definition of $w_{k}$ and $\epsilon_{k}$ we have
	$$
	\|(-\Delta)^{1/4}w_{k}\|_{L^{2}}^{2}=\epsilon_{k}\|(-\Delta)^{1/4}u_{k}\|_{L^{2}}^{2}=1,
	$$
	which implies that $w_{k}$ is bounded in $H^{1/2}(\mathbb{R}^3)$. Hence, we deduce from Lemma \ref{lem:positive ground state solution} that $w_{k}\to w$ strongly in $H^{1/2}(\mathbb{R}^3)$, where either $w\equiv0$ or $w>0$ in $\mathbb{R}^{3}$. Extracting a subsequence if necessary, we have $w_{k}\to w$ strong in $L_{{\rm loc}}^{r}(\mathbb{R}^3)$ for $2\leq r<3$. Thus it follows from \eqref{eq:vanishing6} that
	$$
	\int_{B(0,R_{0})}\left|w(x)\right|^{2}{\rm d}x=\lim_{k\to\infty}\int_{B(0,R_{0})}\left|w_{k}(x)\right|^{2}{\rm d}x>0.
	$$
	This implies that $w\not\equiv0$, and hence $w>0$. Hence, up to translations, we have $w(x)=Q(x)$ where $Q\in \mathcal{G}$.
\end{proof}

\begin{lem}\label{lem:V/epsilon periodic} For any $0<p\leq 1$, there exists a constant $C_{0}>0$ independent of $\left\{ a_{k}\right\}$ such that 
	\begin{equation}\label{ineq: refine V periodic}
	\lim_{k\to\infty}\frac{1}{\epsilon_{k}^{p}}\int_{\mathbb{R}^{3}}V(\epsilon_{k}x+\epsilon_{k}y_{k})\left|w_{k}(x)\right|^{2}{\rm d}x\geq C_{0}.
	\end{equation}
\end{lem}
\begin{proof}	
	We first show that $\frac{x_{k}-x_{0}}{\epsilon_{k}}$ is uniformly
	bounded as $k\to\infty$. On contrary, we assume that there exists a subsequence of $\left\{ a_{k}\right\}$, still denoted by $\left\{ a_{k}\right\}$, such that $\left|\frac{x_{k}-x_{0}}{\epsilon_{k}}\right|\to\infty$ as $k\to\infty$. By Fatou's Lemma we have, for any $0<p\leq 1$ and large constant $M>0$,
	\begin{align*}
	& \lim_{k\to\infty}\frac{1}{\epsilon_{k}^{p}}\int_{\mathbb{R}^{3}}V(\epsilon_{k}x+\epsilon_{k}y_{k})\left|w_{k}(x)\right|^{2}{\rm d}x\geq\int_{\mathbb{R}^{3}}\lim_{k\to\infty}\frac{1}{\epsilon_{k}^{p}}V(\epsilon_{k}x+x_{k})\left|w_{k}(x)\right|^{2}{\rm d}x\\
	& \geq\int_{\mathbb{R}^{3}}\lim_{k\to\infty}\kappa\left|x+\frac{x_{k}-x_{0}}{\epsilon_{k}}\right|^{p}\left|Q(x)\right|^{2}{\rm d}x\geq M.
	\end{align*}
	From \eqref{eq:boson star inequality} and the fact that $\epsilon_{k}^{-3/2}w_{k}(\epsilon_{k}^{-1}x-y_{k})=u_{k}(x)$ is a minimizer of $I(a_{k})$, we have
	\begin{align*}
	I(a_{k}) & =\frac{1}{\epsilon_{k}}\left(\|(-\Delta+m^{2}\epsilon_{k}^{2})^{1/4}w_{k}\|_{L^{2}}^{2}-\frac{a^{*}}{2}\iint_{\mathbb{R}^{3}\times\mathbb{R}^{3}}\frac{\left|w_{k}(x)\right|^{2}\left|w_{k}(y)\right|^{2}}{\left|x-y\right|}{\rm d}x{\rm d}y\right)\\
	& \quad+\frac{a^{*}-a_{k}}{2\epsilon_{k}}\iint_{\mathbb{R}^{3}\times\mathbb{R}^{3}}\frac{\left|w_{k}(x)\right|^{2}\left|w_{k}(y)\right|^{2}}{\left|x-y\right|}{\rm d}x{\rm d}y+\int_{\mathbb{R}^{3}}V(\epsilon_{k}x+z_{k})\left|w_{k}(x)\right|^{2}{\rm d}x\\
	&\geq C_{1}\frac{a^{*}-a_{k}}{\epsilon_{k}}+M\epsilon_{k}^{p}\geq C_{2}M^{\frac{1}{p+1}}\left(a^{*}-a_{k}\right)^{\frac{p}{p+1}}.
	\end{align*}
	The above lower bound of $I(a_{k})$ holds true for any $0<p\leq 1$ and for an arbitrary large $M>0$. This contradicts the upper bound of $I(a_{k})$ in \eqref{ineq: upper bound I(a_{k}) periodic bad}.
	Thus $\frac{x_{k}-x_{0}}{\epsilon_{k}}$ is uniformly bounded, and hence there exists a constant $C_{3}$ such that $\frac{x_{k}-x_{0}}{\epsilon_{k}}\to C_{3}$
	as $k\to\infty$. Hence, by Fatou's Lemma we obtain, for $0<p\leq 1$, 
	\begin{align*}
	&
	\lim_{k\to\infty}\frac{1}{\epsilon_{k}^{p}}\int_{\mathbb{R}^{3}}V(\epsilon_{k}x+\epsilon_{k}y_{k})\left|w_{k}(x)\right|^{2}{\rm d}x \geq\int_{\mathbb{R}^{3}}\lim_{k\to\infty}\frac{1}{\epsilon_{k}^{p}}V(\epsilon_{k}x+x_{k})\left|w_{k}(x)\right|^{2}{\rm d}x\\
	&
	\geq\int_{\mathbb{R}^{3}}\lim_{k\to\infty}\kappa\left|x+\frac{x_{k}-x_{0}}{\epsilon_{k}}\right|^{p}\left|Q(x)\right|^{2}{\rm d}x\geq \int_{\mathbb{R}^{3}}\kappa\left|x+C_{3}\right|^{p}\left|Q(x)\right|^{2}{\rm d}x=C_{0},
	\end{align*}
	for some constant $C_{0}>0$ independent of $a_{k}$. 
\end{proof}

Now we are able to improve the lower bound of $I(a_{k})$ in \eqref{ineq: upper bound I(a_{k}) periodic bad}.
\begin{lem}\label{lem: estimate I(a_{k}) periodic}	There exist positive constants $M_{1}<M_{2}$ independent of $a_{k}$ such that 
	\begin{equation}\label{ineq: upper bound I(a_{k}) periodic}
	M_{1}(a^{*}-a_{k})^{\frac{q}{q+1}}\leq I(a_{k})\leq M_{2}(a^{*}-a_{k})^{\frac{q}{q+1}}.
	\end{equation}
\end{lem}
\begin{proof}
	By \eqref{ineq: upper bound I(a_{k}) periodic bad}, we only need to prove the lower bound in \eqref{ineq: upper bound I(a_{k}) periodic}. Since $I(a_k)$ is decreasing and uniformly bounded for $0\leq a_k\leq a^*$, it suffices to consider the case when $a_k$ is close to $a^*$. As usual, the proof is divided into two case.
	\begin{itemize}
		\item If $0<p\leq 1$, then $q=p$. By the same argument of the proof of \eqref{ineq: refine V periodic} we have 
		$$
		I(a_{k})\geq C_{1}\frac{a^{*}-a_{k}}{\epsilon_{k}}+C_{0}\epsilon_{k}^{p}\geq C_{2}C_{0}^{\frac{1}{p+1}}\left(a^{*}-a_{k}\right)^{\frac{p}{p+1}} = M_{1} \left(a^{*}-a_{k}\right)^{\frac{p}{p+1}}.
		$$
	\end{itemize}
	
	\begin{itemize}
		\item If $p>1$, then $q=1$. By the same arguments of proof of lower bound in \eqref{ineq:estimate I(a_{k})}, we arrive at the lower bound in \eqref{ineq: upper bound I(a_{k}) periodic}.
	\end{itemize}
\end{proof}

By Lemma \ref{lem: estimate I(a_{k}) periodic} and using a similar procedures as the proof of Lemma \ref{lem:estimate of direct}, we obtain the
following estimates for minimizers of $I(a_{k})$.
\begin{lem}\label{lem: estimate of direct term periodic} There exist positive constants $K_{1}<K_{2}$ independent of $a_{k}$ such that 
	$$
	K_{1}(a^{*}-a_{k})^{-\frac{1}{q+1}}\leq\iint_{\mathbb{R}^{3}\times\mathbb{R}^{3}}\frac{|u(x)|^{2}|u(y)|^{2}}{\left|x-y\right|}{\rm d}x{\rm d}y\leq K_{2}(a^{*}-a_{k})^{-\frac{1}{q+1}}.
	$$
\end{lem}

Now let $\epsilon_{k}:=(a^{*}-a_{k})^{\frac{1}{q+1}}>0$, we see that $\epsilon_{k}\to 0$ as $k\to\infty$. We define $\tilde{w}_{k}(x):=\epsilon_{k}^{3/2}u_{k}(\epsilon_{k}x)$ be $L^{2}$--normalized of $u_{k}$. It follows from Lemma \ref{lem: estimate of direct term periodic} that
\begin{equation}\label{ineq:estimate direct term of w_k-0}
0<K_{1}\leq\iint_{\mathbb{R}^{3}\times\mathbb{R}^{3}}\frac{\left|\tilde{w}_{k}(x)\right|^{2}\left|\tilde{w}_{k}(y)\right|^{2}}{\left|x-y\right|}{\rm d}x{\rm d}y\leq K_{2}.
\end{equation}

By the same arguments of proof of Lemma \ref{lem: important periodic} (ii), we can prove (in replacing contradiction \eqref{eq:convergence of direct term and epsilon 1} by \eqref{ineq:estimate direct term of w_k-0})
that there exist sequence $\left\{ y_{k}\right\}\subset\mathbb{R}^{3}$ and positive constant $R_{1}$ such that 
\begin{equation}\label{eq:vanishing7}
\liminf_{k\to\infty}\int_{B(y_{k},R_{1})}\left|\tilde{w}_{k}(x)\right|^{2}{\rm d}x>0.
\end{equation}
Moreover, we can write $\epsilon_{k}y_{k}=z_{k}+x_{k}$ such that $z_{k}\in\mathbb{Z}^3$ and $x_{k}\in [0,1]^3$.

By the same arguments of proof of Lemma \ref{lem:V/epsilon periodic}, we can prove that there exists a constant $C_{0}$ such that $\frac{x_{k}-x_{0}}{\epsilon_{k}}\to C_{0}$
as $k\to\infty$. Thus, by Fatou's Lemma we have, for $0<p\leq 1$,
\begin{align*}
&
\lim_{k\to\infty}\frac{1}{\epsilon_{k}^{p}}\int_{\mathbb{R}^{3}}V(\epsilon_{k}x+\epsilon_{k}y_{k})\left|w_{k}(x)\right|^{2}{\rm d}x \geq\int_{\mathbb{R}^{3}}\lim_{k\to\infty}\kappa\left|x+\frac{x_{k}-x_{0}}{\epsilon_{k}}\right|^{p}\left|w_{k}(x)\right|^{2}{\rm d}x\\
&
=\int_{\mathbb{R}^{3}}\kappa\left|x+C_{0}\right|^{p}\left|w(x)\right|^{2}{\rm d}x=\frac{\kappa}{\lambda^{p}}\int_{\mathbb{R}^{3}}\left|x+C_{0}\lambda\right|^{p}\left|Q(x)\right|^{2}{\rm d}x\\
& \geq\frac{\kappa}{\lambda^{p}}\int_{\mathbb{R}^{3}}\left|x\right|^{p}\left|Q(x)\right|^{2}{\rm d}x,
\end{align*}
since $Q\in \mathcal{G}$ is a radial decreasing function.

By the same arguments of proof of Theorem \ref{thm:behavior-trapping} we can prove that the following strongly convergences hold true in $H^{1/2}(\mathbb{R}^3)$.
 \begin{itemize}
	\item If $0<p\leq 1$, then
	$$
	\lim_{k\to\infty}\left(a^{*}-a_{k}\right)^{\frac{3}{2\left(p+1\right)}}u_{a_{k}}\left(x_{0}+z_{k}+x\left(a^{*}-a_{k}\right)^{\frac{1}{p+1}}\right)=\lambda^{\frac{3}{2}}Q\left(\lambda x\right)
	$$
	
	\item If $p> 1$, then there exists a sequence $\{x_{k}\}\subset [0,1]^3$ such that
	$$
	\lim_{k\to\infty}\left(a^{*}-a_{k}\right)^{\frac{3}{4}}u_{a_{k}}\left(x_{k}+z_{k}+x\left(a^{*}-a_{k}\right)^{\frac{1}{2}}\right)=\lambda^{\frac{3}{2}}Q\left(\lambda x\right)
	$$
\end{itemize}
where $\lambda$ is determined similarly as in Theorem \ref{thm:behavior-trapping}.

\subsection{Proof of Theorem \ref{thm:behavior-ring-shaped}}

Let $p>0$, and we define $q=\min\{p,1\}$.

First, it follows from the fact that $||x|-1|\leq |x-1|$ and the same argument of the proof of upper bound in \eqref{ineq:estimate I(a_{k})}, that there exists a constant $M_{2}>0$ such that 
\begin{equation}\label{upper bound I(a_{k}) ring-shape}
I(a_{k})\leq M_{2}(a^{*}-a_{k})^{\frac{q}{q+1}}.
\end{equation}

From \eqref{eq:boson star inequality} we have 
$$
0\leq \|(-\Delta)^{1/4}u_{k}\|_{L^{2}}^{2}-\frac{a_{k}}{2}\iint_{\mathbb{R}^{3}\times\mathbb{R}^{3}}\frac{\left|u_{k}(x)\right|^{2}\left|u_{k}(y)\right|^{2}}{\left|x-y\right|}{\rm d}x{\rm d}y\leq I(a_{k}).
$$
Thus, we deduce from the upper bound of $I(a_{k})$ in
\eqref{upper bound I(a_{k}) ring-shape} that 
\begin{equation}\label{eq:convergence of direct term and epsilon}
\lim_{k\to\infty}\|(-\Delta)^{1/4}u_{k}\|_{L^{2}}^{-2}\iint_{\mathbb{R}^{3}\times\mathbb{R}^{3}}\frac{\left|u_{k}(x)\right|^{2}\left|u_{k}(y)\right|^{2}}{\left|x-y\right|}{\rm d}x{\rm d}y=\frac{2}{a^{*}}.
\end{equation}
\begin{lem}\label{lem: important ring-shape} We define $\epsilon_{k}^{-1}:=\|(-\Delta)^{1/4}u_{k}\|_{L^{2}}^{2}$. Then the following statements hold true
	\begin{itemize}
		\item[(i)] $\epsilon_{k}\to0$ as $k\to\infty$.
		
		\item[(ii)] There exist a sequence $\left\{ y_{k}\right\}\subset\mathbb{R}^{3}$
		and a postitive constant $R_{2}$ such that the sequence $w_{k}(x)=\epsilon_{k}^{3/2}u_{k}\left(\epsilon_{k}x+\epsilon_{k}y_{k}\right)$ satisfies 
		\begin{equation}\label{eq:vanishing1}
		\liminf_{k\to\infty}\int_{B(0,R_{2})}\left|w_{k}(x)\right|^{2}{\rm d}x>0.
		\end{equation}
		
		\item[(iii)] Let $x_{k}:=\epsilon_{k}y_{k}$, then the sequence $\left\{ x_{k}\right\} $ is bounded uniformly for $k\to\infty$. Moreover $x_{k}\to x_{0}$ as $k\to\infty$, for some $x_{0}\in\mathbb{R}^3$ such that $\left|x_{0}\right|=1$.
	\end{itemize}
	
	Futhermore, we have $w_{k}(x)\to Q(x)$
	strongly in $H^{1/2}(\mathbb{R}^{3})$ as $k\to\infty$, where $Q\in \mathcal{G}$. 
\end{lem}
\begin{proof}
	(i) On the contrary, we assume that there exists subsequence of $\left\{ a_{k}\right\} $,
	still denoted by $\left\{ a_{k}\right\} $, such that $\left\{ u_{k}\right\} $ is bounded in $H^{1/2}(\mathbb{R}^{3})$.
	We infer from Lemma \ref{lem:compact embedding} that there exist $u\in H^{1/2}(\mathbb{R}^{3})$
	such that $u_{k}\rightharpoonup u$ weakly in $H^{1/2}(\mathbb{R}^{3})$
	and $u_{k}\to u$ strongly in $L^{r}(\mathbb{R}^{3})$ for $2\leq r<3$.
	By the Hardy-Littlewood-Sobolev inequality we have 
	$$
	\lim_{k\to\infty}\iint_{\mathbb{R}^{3}\times\mathbb{R}^{3}}\frac{\left|u_{k}(x)\right|^{2}\left|u_{k}(y)\right|^{2}}{\left|x-y\right|}{\rm d}x{\rm d}y=\iint_{\mathbb{R}^{3}\times\mathbb{R}^{3}}\frac{|u(x)|^{2}|u(y)|^{2}}{\left|x-y\right|}{\rm d}x{\rm d}y.
	$$
	Consequently, we have 
	$$
	0=I(a^{*})\leq\mathcal{E}_{a^{*}}\left(u\right)\leq\lim_{k\to\infty}\mathcal{E}_{a_{k}}\left(u_{k}\right)=\lim_{k\to\infty}I(a_{k})=0.
	$$
	This indicates that $u$ is a minimizer of $I(a^{*})$,
	which contradicts to Theorem \ref{thm:existence of minimizers} (ii).
	Thus the claim (i) holds true.
	
	(ii) The proof of \eqref{eq:vanishing1} follows from the same arguments of proof of \eqref{eq:vanishing6} in Lemma \ref{lem: important periodic} (ii).	
	
	(iii) On contrary, we suppose $\left\{ x_{k}\right\} $ is unbounded as $k\to\infty$. Then	there exists a subsequence of $\left\{ a_{k}\right\} $, still denoted	by $\left\{ a_{k}\right\} $, such that $|x_{k}|\to\infty$ as $k\to\infty$.	Since $V(x)\to\infty$ as $\left|x\right|\to\infty$, we derive from	\eqref{eq:vanishing1} and Fatou's Lemma that for any $C>0$,
	$$
	\lim_{k\to\infty}\int_{\mathbb{R}^{3}}V(\epsilon_{k}x+x_{k})\left|w_{k}(x)\right|^{2}{\rm d}x\geq C>0
	$$
	which is impossible since 
	$$
	0\leq \int_{\mathbb{R}^{3}}V\left(\epsilon_{k}x+x_{k}\right)\left|w_{k}(x)\right|^{2}{\rm d}x=\int_{\mathbb{R}^{3}}V(x)\left|u_{k}(x)\right|^{2}{\rm d}x\leq I(a_{k})
	$$
	and $I(a_{k})\to0$ as $k\to \infty$, thanks to the uppper bound of $I(a_{k})$ in \eqref{upper bound I(a_{k}) ring-shape}.
	
	Thus $\left\{ x_{k}\right\} $
	is bounded uniformly for $k\to\infty$. Therefore, for any sequence
	$\left\{ a_{k}\right\}$ there exists a subsequence of $\left\{ a_{k}\right\}$, still denoted by $\left\{ a_{k}\right\} $, such that $x_{k}\to x_{0}$ as $k\to\infty$ for some point $x_{0}\in\mathbb{R}^{3}$.
	We claim that $\left|x_{0}\right|=1$. Otherwise, if $\left|x_{0}\right|\ne1$ then $V(x_{0})>0$. We deduce from \eqref{eq:vanishing1} and Fatou's Lemma that 
	$$
	\lim_{k\to\infty}\int_{\mathbb{R}^{3}}V(\epsilon_{k}x+x_{k})\left|w_{k}(x)\right|^{2}{\rm d}x\geq V(x_{0})\lim_{k\to\infty}\int_{\mathbb{R}^{3}}\left|w_{k}(x)\right|^{2}{\rm d}x>0
	$$
	which is clearly impossible as we have seen before. Thus $\left|x_{0}\right|=1$.
	
	Finally, by the same arguments of proof in Lemma \ref{lem: important periodic}, we can show that $w_{k}\to Q$ strongly in $H^{1/2}(\mathbb{R}^3)$ where $Q\in \mathcal{G}$.
\end{proof}
\begin{lem}
	\label{lem:V/epsilon ring-shape}For any $0<p\leq 1$, there exists a constant $C_{0}>0$ independent of $\left\{ a_{k}\right\}$	such that 	\begin{equation}\label{ineq:refine V ring-shape}
	\lim_{k\to\infty}\frac{1}{\epsilon_{k}^{p}}\int_{\mathbb{R}^{3}}V(\epsilon_{k}x+x_{k})\left|w_{k}(x)\right|^{2}{\rm d}x\geq C_{0}.
	\end{equation}
\end{lem}
\begin{proof} We first show that $\frac{\left|x_{k}\right|-1}{\epsilon_{k}}$	is uniformly bounded as $k\to\infty$. On contrary, we assume that there exists a subsequence of $\left\{ a_{k}\right\}$, still denoted by $\left\{ a_{k}\right\}$, such that $\left|\frac{\left|x_{k}\right|-1}{\epsilon_{k}}\right|\to\infty$ as $k\to\infty$. By Fatou's Lemma we have, for any $0<p\leq 1$ and large constant $C>0$,
	\begin{align*}
	& \lim_{k\to\infty}\frac{1}{\epsilon_{k}^{p}}\int_{\mathbb{R}^{3}}V(\epsilon_{k}x+x_{k})\left|w_{k}(x)\right|^{2}{\rm d}x\geq\int_{\mathbb{R}^{3}}\lim_{k\to\infty}\frac{1}{\epsilon_{k}^{p}}\left|\left|\epsilon_{k}x+x_{k}\right|-1\right|^{p}\left|w_{k}(x)\right|^{2}{\rm d}x\\
	& =\int_{\mathbb{R}^{3}}\lim_{k\to\infty}\left|\left|x+\frac{x_{k}}{\epsilon_{k}}\right|-\frac{1}{\epsilon_{k}}\right|^{p}\left|Q(x)\right|^{2}{\rm d}x\geq C.
	\end{align*}
	From \eqref{eq:boson star inequality} and the fact that $\epsilon_{k}^{-3/2}w_{k}(\epsilon_{k}^{-1}x-y_{k})=u_{k}(x)$ is a minimizer of $I(a_{k})$, we have
	\begin{align*}
	I(a_{k}) & =\frac{1}{\epsilon_{k}}\left(\|(-\Delta+m^{2}\epsilon_{k}^{2})^{1/4}w_{k}\|_{L^{2}}^{2}-\frac{a^{*}}{2}\iint_{\mathbb{R}^{3}\times\mathbb{R}^{3}}\frac{\left|w_{k}(x)\right|^{2}\left|w_{k}(y)\right|^{2}}{\left|x-y\right|}{\rm d}x{\rm d}y\right)\\
	& \quad+\frac{a^{*}-a_{k}}{2\epsilon_{k}}\iint_{\mathbb{R}^{3}\times\mathbb{R}^{3}}\frac{\left|w_{k}(x)\right|^{2}\left|w_{k}(y)\right|^{2}}{\left|x-y\right|}{\rm d}x{\rm d}y+\int_{\mathbb{R}^{3}}V(\epsilon_{k}x+x_{k})\left|w_{k}(x)\right|^{2}{\rm d}x\\
	&\geq C_{1}\frac{a^{*}-a_{k}}{\epsilon_{k}}+C\epsilon_{k}^{p}\geq C_{2}C^{\frac{1}{p+1}}\left(a^{*}-a_{k}\right)^{\frac{p}{p+1}}.
	\end{align*}
	Since the above lower bound of $I(a_{k})$ holds true for any $0<p\leq 1$ and for an arbitrary large $C>0$, the above estimates contradicts the
	upper bound of $I(a_{k})$ in \eqref{upper bound I(a_{k}) ring-shape}.
	Thus $\frac{\left|x_{k}\right|-1}{\epsilon_{k}}$ is uniformly bounded, and hence there exists a constant $C_{3}$ such that $\frac{\left|x_{k}\right|-1}{\epsilon_{k}}\to C_{3}$
	as $k\to\infty$. Thus, by Fatou's Lemma again we obtain, for any $0<p\leq 1$,
	\begin{align*}
	& \lim_{k\to\infty}\frac{1}{\epsilon_{k}^{p}}\int_{\mathbb{R}^{3}}V(\epsilon_{k}x+x_{k})\left|w_{k}(x)\right|^{2}{\rm d}x\\
	\geq & \int_{\mathbb{R}^{3}}\lim_{k\to\infty}\left|\frac{\left|\epsilon_{k}x+x_{k}\right|-1}{\epsilon_{k}}+\frac{\left|x_{k}\right|-1}{\epsilon_{k}}\right|^{p}\left|w_{k}(x)\right|^{2}{\rm d}x\\
	= & \int_{\mathbb{R}^{3}}\left|x_{0}\cdot x+C_{3}\right|^{p}\left|Q(x)\right|^{2}{\rm d}x=C_{0},
	\end{align*}
	for some constant $C_{0}>0$ independent of $a_{k}$. 
\end{proof}

We are now able to establish the
lower bound of $I(a_{k})$ in \eqref{upper bound I(a_{k}) ring-shape}.
\begin{lem}\label{lem: estimate I(a_{k}) ring-shape}
	There exist positive constants $M_{1}<M_{2}$ independent of $a_{k}$ such that 
	\begin{equation}\label{ineq: upper bound I(a_{k}) ring-shape}
	M_{1}(a^{*}-a_{k})^{\frac{q}{q+1}}\leq I(a_{k})\leq M_{2}(a^{*}-a_{k})^{\frac{q}{q+1}}.
	\end{equation}
\end{lem}
\begin{proof}
	By \eqref{upper bound I(a_{k}) ring-shape}, we only need to prove the lower bound in \eqref{ineq: upper bound I(a_{k}) ring-shape}. Since $I(a_k)$ is decreasing and uniformly bounded for $0\leq a_k\leq a^*$, it suffices to consider the case when $a_k$ is close to $a^*$. As usual, the proof is divided into two case.
	\begin{itemize}
		\item If $0<p\leq 1$, then $q=p$. By the same argument of the proof of \eqref{ineq:refine V ring-shape} we have 
		$$
		I(a_{k})\geq C_{1}\frac{a^{*}-a_{k}}{\epsilon_{k}}+C_{0}\epsilon_{k}^{p}\geq C_{2}C_{0}^{\frac{1}{p+1}}\left(a^{*}-a_{k}\right)^{\frac{p}{p+1}} = M_{1} \left(a^{*}-a_{k}\right)^{\frac{p}{p+1}}.
		$$
	\end{itemize}
	
	\begin{itemize}
		\item If $p>1$, then $q=1$. By the same arguments of proof of lower bound in \eqref{ineq:estimate I(a_{k})}, we arrive at the lower bound in \eqref{ineq: upper bound I(a_{k}) ring-shape}.
	\end{itemize}
\end{proof}

By Lemma \ref{lem: estimate I(a_{k}) ring-shape} and using a similar procedures as the proof of Lemma 4 in \cite{Ng-17}, we obtain the
following estimates for minimizers of $I(a_{k})$.
\begin{lem}\label{lem: estimate of direct term ring-shape} There exist positive constants $K_{1}<K_{2}$ independent of $a_{k}$ such that 
	$$
	K_{1}(a^{*}-a_{k})^{-\frac{1}{q+1}}\leq\iint_{\mathbb{R}^{3}\times\mathbb{R}^{3}}\frac{\left|u_{k}(x)\right|^{2}\left|u_{k}(y)\right|^{2}}{\left|x-y\right|}{\rm d}x{\rm d}y\leq K_{2}(a^{*}-a_{k})^{-\frac{1}{q+1}}.
	$$
\end{lem}

Now let $\epsilon_{k}:=(a^{*}-a_{k})^{\frac{1}{q+1}}>0$, we see that $\epsilon_{k}\to 0$ as $k\to\infty$. We define $\tilde{w}_{k}(x):=\epsilon_{k}^{3/2}u_{k}(\epsilon_{k}x)$ be $L^{2}$--normalized of $u_{k}$. It follows from Lemma \ref{lem: estimate of direct term ring-shape} that
\begin{equation}\label{ineq:estimate direct term of w_k}
0<K_{1}\leq\iint_{\mathbb{R}^{3}\times\mathbb{R}^{3}}\frac{\left|\tilde{w}_{k}(x)\right|^{2}\left|\tilde{w}_{k}(y)\right|^{2}}{\left|x-y\right|}{\rm d}x{\rm d}y\leq K_{2}.
\end{equation}

By the same arguments of proof of Lemma \ref{lem: important periodic} (ii), we can prove (in replacing contradiction \eqref{eq:convergence of direct term and epsilon 1} by \eqref{ineq:estimate direct term of w_k})
that there exist sequence $\left\{ y_{k}\right\}\subset\mathbb{R}^{3}$ and positive constant $R_{2}$ such that 
\begin{equation}\label{eq:vanishing3}
\liminf_{k\to\infty}\int_{B(y_{k},R_{2})}\left|\tilde{w}_{k}(x)\right|^{2}{\rm d}x>0,
\end{equation}

By the same arguments of proof of Lemma \ref{lem:V/epsilon ring-shape}, we can prove that there exists a constant $C_{0}$ such that $\frac{\left|x_{k}\right|-1}{\epsilon_{k}}\to C_{0}$
as $k\to\infty$. Thus, by Fatou's Lemma we have, for $0<p\leq 1$,
\begin{align*}
& \lim_{k\to\infty}\frac{1}{\epsilon_{k}^{p}}\int_{\mathbb{R}^{3}}V\left(\epsilon_{k}x+x_{k}\right)\left| w_{k}(x)\right|^{2}{\rm d}x \\
& \geq\int_{\mathbb{R}^{3}}\lim_{k\to\infty}\left|\frac{\left|x\epsilon_{k}+x_{k}\right|-\left|x_{k}\right|}{\epsilon_{k}}+\frac{\left|x_{k}\right|-1}{\epsilon_{k}}\right|^{p}\left| w_{k}(x)\right|^{2}{\rm d}x \\
& =\int_{\mathbb{R}^{3}}\left| x_{0}\cdot x+C_{0}\right|^{p}\left| w(x)\right|^{2}{\rm d}x =\frac{1}{\lambda^{p}}\int_{\mathbb{R}^{3}}\left| x_{0}\cdot x + C_{0}\lambda\right|^{p}\left| Q(x)\right|^{2}{\rm d}x\\
& \geq\frac{1}{\lambda^{p}}\int_{\mathbb{R}^{3}}\left| x_{0}\cdot x\right|^{p}\left|Q(x)\right|^{2}{\rm d}x,
\end{align*}
since $Q\in \mathcal{G}$ is a radial decreasing function.

By the same arguments of proof of Theorem \ref{thm:behavior-trapping} we can prove that 
$$
\lim_{k\to\infty}\left(a^{*}-a_{k}\right)^{\frac{3}{2\left(q+1\right)}}u_{a_{k}}\left(x_k+x\left(a^{*}-a_{k}\right)^{\frac{1}{q+1}}\right)=\lambda^{\frac{3}{2}}Q(\lambda x)
$$
strongly in $H^{1/2}(\mathbb{R}^{3})$, where $\lambda$ is determined as in Theorem \ref{thm:behavior-ring-shaped}.

\section*{Acknowledgement}
The author would like to thank P.T. Nam for helpful discussions.

\end{document}